\newif\ifshort
\shortfalse

\documentclass[sigconf]{aamas}

\usepackage{balance} %

\usepackage{framed}
\usepackage{algorithm}
\usepackage[noend]{algpseudocode}
\usepackage[title]{appendix}
\usepackage{thm-restate}
\usepackage{pict2e}
\usepackage{colortbl}
\usepackage{xcolor}
\definecolor{lightblue}{RGB}{119,170,221} 
\definecolor{orange}{RGB}{238,136,102}
\definecolor{pear}{RGB}{187,204,51} 
\definecolor{lightcyan}{RGB}{153,221,255} 
\definecolor{lightyellow}{RGB}{238,221,136} 
\definecolor{palegrey}{RGB}{221,221,221}
\definecolor{pink}{RGB}{255,170,187}
\definecolor{olive}{RGB}{170,170,0}
\definecolor{mint}{RGB}{68,187,153}
\definecolor{darkgrey}{RGB}{80,80,80}  

\renewcommand\footnotetextcopyrightpermission[1]{\footnote{This manuscript is the full version of the paper accepted at the \emph{25th International Conference on Autonomous Agents and Multiagent Systems (AAMAS 2026).}}}

\title{Necessary President in Elections with Parties}

\author{Katarína Cechlárová}
\affiliation{
  \institution{P.J. \v Saf\'arik University}
  \city{Ko\v sice}
  \country{Slovakia}}
\email{katarina.cechlarova@upjs.sk}

\author{Ildikó Schlotter}
\affiliation{
  \institution{ELTE Centre for Economic and Regional Studies}
  \city{Budapest}
  \country{Hungary}}
\affiliation{
  \institution{Budapest University of Technology and Economics}
  \city{Budapest}
  \country{Hungary}}
\email{schlotter.ildiko@krtk.elte.hu}

\begin{abstract}
Consider an election where the set of candidates is partitioned into parties, and each party must choose exactly one candidate to nominate for the election held over all nominees. The \textsc{Necessary President} problem asks whether a candidate, if nominated, becomes the  winner of the election for all possible  nominations from other parties. 

We study the computational complexity of \textsc{Necessary President} for several voting rules. We show that while this problem is solvable in polynomial time for Borda, Maximin, and Copeland$^\alpha$ for every $\alpha\in [0,1]$, 
it is $\mathsf{coNP}$-complete for general classes of positional scoring rules that include $\ell$-Approval and $\ell$-Veto, even when the maximum size of a party is two. For such positional scoring rules, we show that \textsc{Necessary President} is
 $\mathsf{W}[2]$-hard when  parameterized by  the number of parties, but fixed-parameter tractable with respect to the number of voter types.
Additionally, we prove that \textsc{Necessary President}
for Ranked Pairs is $\mathsf{coNP}$-complete even for maximum party size two, and 
$\mathsf{W}[1]$-hard with respect to the number of parties; remarkably, both of these results hold even for constant number of voters.
\end{abstract}


\newcommand{\myref}[1]{%
\ifshort%
NP-Copeland%
\else%
\ref{#1}%
\fi%
}

\newcommand{\new}[1]{{#1}}
\def\opp{w}
\def\oppParty{P_w}

\algrenewcommand\algorithmicrequire{\textbf{Input:}}
\newenvironment{varalgorithm}[1]  {\algorithm\renewcommand{\thealgorithm}{#1}}
  {\endalgorithm}
  
\newcommand{\linkproof}[1]{%
\ifshort
    $\star$%
\else
    \hyperref[#1]{$\star$}%
\fi 
}

\newcommand{\opentriangle}{%
  \raisebox{0.2pt}{\makebox[0.77778em]{%
    \setlength{\unitlength}{0.6em}%
    \linethickness{0.4pt}\roundjoin
    \begin{picture}(1,1)
    \polygon(0,0)(1,0)(1,1)
    \end{picture}%
  }}%
}

\usepackage[textsize=small,textwidth=4cm,shadow]{todonotes}

\def\NP{\mathsf{NP}}
\def\poly{\mathsf{P}}
\def\coNP{\mathsf{coNP}}
\def\W1{\mathsf{W}[1]}
\def\W2{\mathsf{W}[2]}
\def\FPT{\mathsf{FPT}}
\def\paraNP{\textup{para-}\mathsf{NP}}
\def\XP{\mathsf{XP}}

\newtheorem{observation}{Observation}

\allowdisplaybreaks

\def\NPP{\textsc{Necessary President}}
\def\MM{\mathsf{MM}}

\newcommand{\scr}[1][\mathcal{E}]{\mathsf{scr}_{#1}}
\newcommand{\Cpl}[1][\mathcal{E}]{\mathsf{Cpl}_{#1}^{\alpha}}
\newcommand{\ora}[1]{\overrightarrow{#1}}
\newcommand{\ola}[1]{\overleftarrow{#1}}
\newcommand{\ol}[1]{\overline{#1}}
\def\mysim{\overset \E\sim}
\def\mysim{\sim}
\def\place{\mathsf{prt}}
\def\maxsize{s}

\newcommand{\shortdots}{\scalebox{0.8}[1]{$\!\cdots\!$}}

\def\E{\mathcal{E}}
\def\C{\mathcal{C}}
\def\P{\mathcal{P}}
\def\Q{\mathcal{Q}}
\def\S{\mathcal{S}}

\def\R{\mathcal{R}}

\begin{document}

\pagestyle{fancy}
\fancyhead{}

\maketitle

\section{Introduction}

Most political elections are preceded by a turbulent and intense period when parties 
want to decide which candidate to nominate for the election. 
The nomination process may take the form of primaries, or may involve more complex, strategic decisions that are not only based on the candidates' traits as viewed by the party members but also on the preferences of voters.
Indeed, as the election approaches, parties may realize that the candidate previously picked by the party---e.g., the winner of a primary---has a low support in the polls when compared to the nominees of other parties, and thus needs to be replaced.

Among recent examples of such strategic nominations, the most famous one is perhaps 
the replacement of Joe Biden by Kamala Harris before the 2024 US presidential election~\cite{2024USelectionswiki}.
In Poland, the largest 
opposition group Civic Platform  replaced her primary election winner  Ma\l gorzata Kidawa-B\l o\'nska  by  Rafa\l{} Trzaskowski after a significant drop in her support before  the  2020 presidential elections~\cite{tvn24-PolishCivicCoalition}.
Notably, in neither case was the substitute candidate able to ensure victory.

Sometimes the main purpose of parties is not victory itself, but rather to prevent a certain candidate from winning, leading to various forms of strategic nomination. Such a process was witnessed in the 2022 Hungarian elections, where opposition parties decided to cooperate and create a joint list of nominees in order to be able to defeat the ruling party Fidesz~\cite{2022Hungarianelectionswiki}.
A similar cooperation was carried out in the 2023 Turkish elections where the opposition formed a six-party alliance to nominate Kemal K{\i}l{\i}\c{c}daro\u{g}lu in the hope of defeating the ruling president Recep Tayyip Erdo\u{g}an  \cite{2023Turkishelectionlemonde}. 
\ifshort
\else
Both efforts were unsuccessful. 
\fi 

Several interesting questions arise in such situations. 
In this paper, we address one of the most basic ones: 
    Can a given nominee~$p$ participating in an election be defeated with a judicious choice of nominations from all remaining parties? Or will $p$ necessarily win, irrespective of the nominees chosen by all other parties?

To study this topic, we use the
formal model of candidate nomination  as introduced by Faliszewski et al.~\cite{faliszewski2016} 
where
parties are interpreted as sets of candidates, and each party has to nominate exactly one of its candidates for the upcoming election. %
Faliszewski et al.\ assumed 
that parties know the preferences of {\it all} voters over {\it all} potential candidates, 
and
studied two problems in this setting: The {\sc Possible President}  problem asks whether a given party can nominate one of its candidates in such a way that he or she becomes a winner of the election for \emph{some} nominations from other parties, while  the {\sc Necessary President} problem asks if a given nominee of the party will be a winner \emph{irrespective of all other nominations}.

Faliszewski et al. \cite{faliszewski2016} concentrated on Plurality. For this voting rule they  proved that the {\sc Possible President}  problem is $\NP$-complete and the  {\sc Necessary President} problem is $\coNP$-complete if voters' preferences are  unrestricted. Motivated by these hardness results, they focused on structured preferences and showed that {\sc Necessary President} admits a polynomial-time algorithm for single-peaked profiles. By contrast,  {\sc Possible President} remains $\NP$-complete even on 1D-Euclidean profiles but admits a polynomial-time algorithm if the elections are restricted to single-peaked profiles where the candidates of each party appear consecutively on the societal axis.

In contrast to the steadily growing research
concerning the classical and parameterized computational complexity of the {\sc Possible President} problem for a whole range of different voting rules~\cite{misra2019parameterized,lisowski2022strategic,cechlarova2023candidates,schlotter2024,schlotter2025candidate}, not much is known about the {\sc Necessary President} problem. 
 Cechlárová et al.~\cite{cechlarova2023candidates} generalized the intractability result of Faliszewski et al.~\cite{faliszewski2016} by
 showing that \NPP{} is $\coNP$-complete for Plurality with run-off, as well as for the voting rules $\ell$-Approval and  $\ell$-Veto for every positive integer $\ell$; all these results hold in the restricted setting where each party has at most two candidates.
 The authors also provided integer programs for {\sc Necessary President} for a wide range of voting rules and performed computational experiments based on real and synthetic election data.
 Complementing the tractability result by Faliszewski et al.~\cite{faliszewski2016} for \NPP{} for Plurality on single-peaked preference profiles, Misra~\cite{misra2019parameterized} showed that the same problem is polynomial-time solvable on single-crossing profiles.

\subsection{Our Contribution}

\begin{table}[t]
	\caption{Our results for \NPP{}.
    The column `Param./Const.' contains the considered parameters or their restriction to a constant; `eff.~comp.' stands for voting rules where the winners of an election can be computed efficiently, i.e., in polynomial time. For Copeland$^\alpha$, $\alpha$ is in $[0,1]$.}
	\label{tab:summary}
	\begin{tabular}{@{\hspace{2pt}}l@{\hspace{2pt}}c@{\hspace{2pt}}c@{\hspace{8pt}}c@{\hspace{2pt}}}\toprule
		Voting rule & Param./Const. & Complexity & Reference\\ \midrule
		Borda & -- & in $\poly$ 
        & Thm.~\ref{thm:Borda}\\
		Short & $s=2$ & $\coNP$-complete
        & Thm.~\ref{thm:short-coNP-s}\\
        & $t$ & $\mathsf{W}[2]$-hard, in $\XP$
        & Thm.~\ref{thm:short-W2hard-t}, Obs.~\ref{obs:bruteforce}\\
        & $\tau$ & in $\FPT$
        & Thm.~\ref{thm:FPT-short}\\
        Veto-like & $s=2$ & $\coNP$-complete
        & Thm.~\ref{thm:veto-coNP-s}\\
         & $t$ & $\mathsf{W}[2]$-hard, in $\XP$
        & Thm.~\ref{thm:Veto-W2hard-t}, Obs.~\ref{obs:bruteforce}\\
        & $\tau$ & in $\FPT$
        & Thm.~\ref{thm:FPT-veto}\\
        Copeland$^\alpha$ & -- & in $\poly$ 
        & Thm.~\ref{thm:Copeland}\\
		Maximin & -- & in $\poly$ 
        & Thm.~\ref{thm:Maximin}\\
		Ranked Pairs & $s=2$, $|V|=12$ & $\coNP$-complete
        & Thm.~\ref{thm:RPcoNP}\\
         & $t$, $|V|=20$ & $\mathsf{W}[1]$-hard, in $\XP$
        & Thm.~\ref{thm:RP_W1hard}, Obs.~\ref{obs:bruteforce}\\
        eff. comp. & $s,t$ & in $\FPT$ & Obs.~\ref{obs:bruteforce}\\
		\bottomrule
	\end{tabular}
\end{table}

We study the computational complexity of the \NPP{} problem for several positional scoring rules and a variety of Condorcet-consistent voting rules; our results are summarized in Table~\ref{tab:summary}.  
The positional scoring rules we consider include the Borda rule as well as the set of all so-called \emph{short} and \emph{Veto-like} voting rules, introduced by Schlotter et al.~\cite{schlotter2024}. 
Additionally, we consider some of the most well-known Condorcet-consistent voting rules: Copeland$^\alpha$ for all $\alpha \in [0,1]$, Maximin, and Ranked Pairs.

We classify the complexity of \NPP{} for each of these voting rules as either polynomial-time solvable or $\coNP$-complete. Furthermore, we apply the framework of parameterized complexity to deal with the computationally intractable cases: we examine how certain natural parameters of a given instance influence the computational complexity of the \NPP{} problem. 
The parameters we consider are the following: 
\begin{itemize}
    \item $t$: the number of parties;
    \item $s$: the maximum size of a party;
    \item $\tau$: the number of \emph{voter types}, where two voters are of the same type if they have the same preferences over the candidates;
    \item $|V|$: the number of voters.
\end{itemize}

For each of the voting rules for which \NPP{} turns out to be $\coNP$-complete, we settle its parameterized complexity for every possible combination of the above four parameters as either (i) fixed-parameter tractable (FPT), (ii) $\mathsf{W}[1]$- or $\mathsf{W}[2]$-hard and in $\XP$, or (iii) $\paraNP$-hard.

\subsection{Related Work}
\label{sec:relatedwork}
The line of research most closely related to this paper---the problem of strategic candidate nomination by parties preceding an election---was initiated by Faliszewski et al.~\cite{faliszewski2016} and, apart from the papers already mentioned~\cite{faliszewski2016,cechlarova2023candidates}, has focused on the \textsc{Possible President} problem; 
see Section~\ref{sec:pospres}.
Different models describing elections with parties from an algorithmic viewpoint are briefly discussed in Section~\ref{sec:electpart}.
Another related topic is \emph{strategic candidacy} where candidates are
independent and have full power over deciding whether to run for the election or not; 
see Section~\ref{sec:stratcand}.

Candidate nomination can also be seen as part of the broader topic of elections where the set of candidates is not fixed; such a scenario appeared already in the seminal paper by Bartholdi et al.~\cite{bartholdi1992} in the form of control by adding or deleting candidates by an election chair.  
For an overview of various other results dealing with this type of election control, see for example  Faliszewski and Rothe~\cite{Chapter7}, Chen et al.~\cite{chen2017elections} or Erd\'elyi et al.~\cite{erdelyi2021}.

\subsubsection{The \textsc{Possible President} Problem}
\label{sec:pospres}

The results of Faliszewski et al. \cite{faliszewski2016} for \textsc{Possible President} for Plurality voting have been extended to other voting rules.
Lisowski \cite{lisowski2022strategic} dealt with tournament solutions and showed that  {\sc Possible President} for Condorcet rule (which selects the Condorcet winner if it exists, and the empty set otherwise) can be solved in polynomial time but is $\NP$-complete for the Uncovered Set rule.
Cechlárová et al. \cite{cechlarova2023candidates} studied the problem for positional scoring voting rules, among them $\ell$-Approval, $\ell$-Veto, and Borda, and  Condorcet-consistent rules such as Copeland, Llull, and Maximin. In addition, this paper provides integer programs for the {\sc Possible President} as well as  {\sc Necessary President} problem for all the studied voting rules and computational experiments with these integer programs applied to real and synthetic elections.

The parameterized complexity of \textsc{Possible President} was first studied by Misra~\cite{misra2019parameterized} who
strengthened the results of Faliszewski et al. \cite{faliszewski2016} by showing that \textrm{Possible President} for Plurality is $\paraNP$-hard  when parameterized by the size of the largest party even in profiles that are both single-peaked and single-crossing; she also proved that when parameterized by  the number of parties,   the problem is $\mathsf{W}$[1]-hard and in $\XP$ for general preferences but becomes 
$\FPT$ on the 1D-Euclidean domain.

A  detailed multivariate complexity analysis of the \textsc{Possible President} problem for several classes of positional scoring rules has been provided by Schlotter et al.~\cite{schlotter2024} who %
studied the parameterized complexity of these problems with respect to the four parameters studied in this paper.
The same multivariate approach was taken by 
Schlotter and Cechlárová  \cite{schlotter2025candidate} for two types of Condorcet-consistent voting
rules: Copeland$^\alpha$ for every $\alpha\in [0,1]$ and Maximin.

\subsubsection{Elections Involving Parties}
\label{sec:electpart}
Harrenstein et al.~\cite{harrenstein2021hotelling} introduced a model where voters and candidates are both described by their position on 
the real line (modeling the political spectrum), and each party has to choose its nominee from among its potential candidates under the assumption that each voter votes for
the closest nominee. 
The authors showed that a Nash equilibrium (NE) %
is not guaranteed to exist even in a two-party game, and finding a NE %
is $\mathsf{NP}$-complete in general but can be computed in linear time for two %
parties. 

The above model was extended by Deligkas et al.~\cite{deligkas2022parameterized}: they associated each candidate with a cost and studied the  parameterized complexity of
the equilibrium computation problem under several natural parameters
such as
the number of different positions of the candidates, the so-called discrepancy and span of the nominees, and the maximum overlap of the parties.

Harrenstein and Turrini~\cite{harrenstein2022computing} considered district-based elections.  In each district, voters rank the nominated candidates and elect the Plurality winners, and parties have to strategically place their candidates in districts so as to maximize the number of their nominees that get elected.  The authors showed  that deciding the existence of pure NE  
for these games is $\mathsf{NP}$-complete if the party size is bounded by a constant and  $\Sigma^P_2$-complete in general.

Perek et al.~\cite{perek2013} introduced a model where %
voters, \emph{not} candidates, are partitioned into parties, with voters of the same party voting in the same way.
The authors proposed to measure the threat to the so-called leading party~$P$---the party whose favored candidate is the expected winner of the election---by the maximum number of voters who can abandon $P$ for another party without changing the winner of the election, and by the minimal number of voters that must leave~$P$ to ensure that the winner changes. Perek at al.~\cite{perek2013} and in a follow-up paper Guo et al.~\cite{guo_yang2015} studied the computational complexity of these problems for several different voting rules.

\subsubsection{Strategic Candidacy}
\label{sec:stratcand}

In strategic candidacy games, it is assumed that voters as well as candidates have preferences over possible election outcomes. Candidates  strategically choose whether to join or leave the election. Lang et al.~\cite{lang2025strategic} illustrate how this phenomenon occurs in reality in  the 2017 French presidential election where the  centrist candidate Bayrou withdrew to help Macron qualify to the second round (successfully), and green candidate Jadot withdrew to help the socialist candidate
Hamon qualify (not successfully). 
Similar situations are known also in the history of Slovak presidential elections: in February 2019, Robert Mistrík stepped down and supported Zuzana Čaputová~\cite{spectator-Mistrik-withdraws}, in March 2024 the parliamentary vice-chair Andrej Danko withdrew his presidential candidacy in favor of  ex-minister of justice \v Stefan Harabin~\cite{tars-Danko-withdraws}.

Dutta et al. \cite{dutta2001strategic}  consider a framework in which there is a finite set of voters and potential candidates, while  some  of the candidates may also be voters. Each individual,
including candidates, has preferences over the set of all candidates and it is assumed that every candidate prefers herself to all other candidates. The authors examine a two-stage procedure where in a first stage candidates decide on whether or not they will enter the election, and then in a second stage a voting
procedure is implemented to select the winner from the candidates 
who enter. A voting procedure is called  
{\it candidate stable} if it is always a Nash equilibrium (NE for short) for all candidates
to enter.  The authors show that if the sets of
voters and candidates are disjoint, then the only candidate stable voting procedures  are dictatorial procedures. When the set of candidates and voters overlap, then there exist non-dictatorial voting procedures that satisfy candidate stability and unanimity.

Brill and Conitzer~\cite{brill2015strategic}  extend the analysis
to the case when also  the voters act strategically.
First, they study the strategic candidacy with single-peaked
preferences when the voting rule is majority consistent but not Condorcet-consistent. They define several stability notions and show that when candidates play a strong
equilibrium and voters vote truthfully, the outcome will be the
Condorcet winner. For voting by successive elimination they show that computing the candidate stable set is NP-complete.

Lang et al.~\cite{lang2025strategic}
present an analysis of such games for a list of common voting procedures, among them positional scoring rules (including plurality, veto and Borda), Condorcet consistent voting rules (Maximin, uncovered set and Copeland), plurality with run-off and single transferable vote. The studied question  is  whether such games possess a pure strategy NE in which
the outcome is the same as if all candidates run (called {\it genuine equilibra}). A number of negative results is obtained: unless the number of candidates is small there may be games without such stable outcomes, but for the Copeland rule
 the existence of a genuine equilibrium is guaranteed for any number of candidates and for an odd number of voters. Also, 
a strong relationship between equilibra of candidacy games and a form of voting control by adding or removing candidates, where candidates must consent to addition
or deletion, was established.

Other works assume that the preferences of candidates depend not only on the outcome of the election, but take into account also the monetary and reputational costs of running an electoral
campaign.   Obraztsova et al.~\cite{Obraztsova2015} assume that candidates are lazy, meaning that their  additional value from running an electoral campaign is negative (the campaign incurs some cost) while Lang et al.~\cite{Lang2019Keen} assume that this value is  positive (participating in the election gives the candidate an opportunity to advertise his party or political platform and thus raises his profile/reputation). 

Further line of research in Polukarov et al.~\cite{Polukarov2015} and Obraztsova et al.~\cite{Obraztsova2015} studies equilibrium dynamics in candidacy games, where candidates may strategically decide to enter the election or withdraw their candidacy in each iteration of the process until (and if) it converges to a stable state.

\section{Preliminaries}
\label{sec:prelim}
We use the notation $[i]=\{1,2,\dots,i\}$ for each positive integer $i$.

An election $\E=(C,V,\{\succ_v\}_{v \in V})$ consists of a finite set~$C$ of candidates, a finite set~$V$ of voters, and the preferences of voters over the set~$C$ of candidates. %
We assume that the preferences of voter~$v$ are represented by a strict linear order~$\succ_v$  over~$C$ where $c\succ_v c'$ means that voter $v$ \emph{prefers} candidate~$c$ to candidate~$c'$. If two voters have the same preferences, they are said to be of the same \emph{type}; the number of voter types in $V$ will be denoted by $\tau$.

We denote the set of all elections over a set~$C$ of candidates by~${\mathbb{E}}_C$.
A \emph{voting rule} $\R:{\mathbb E}_C\to 2^C$ chooses a set of \emph{winners} of the election. 

We shall also assume that a partition ${\P}=\{P_1,\dots,P_t\}$ of the set~$C$ of candidates  is given; each set $P_j$ is interpreted as a \emph{party} that has to decide about whom among its potential candidates to nominate for the election.
Formally, a \emph{reduced election}  arises after each party has nominated a unique candidate, leading to a set $C'\subseteq C$ of nominees such that $|C'\cap P_j|=1$ for each ${j\in [t]}$. In the reduced election ${\E}_{C'}=(C',V,\{\succ'_v\}_{v \in V})$ 
each voter~$v \in V$ restricts her original preference relation~$\succ_v$ over~$C$ to~$C'$, yielding~$\succ'_v$.

Now we formulate the problem studied in this paper.

\medskip
\noindent
\begin{minipage}{\columnwidth}
\begin{framed}

\noindent {\bf Problem }{\sc Necessary President} for voting rule $\R$.

\noindent {\bf Instance:} A tuple $I=(\E,\mathcal{P},p)$ where $\E=(C,V,\{\succ_v\}_{v \in V})$ is an election with candidate set~$C$ and voter set~$V$,  a partition~${\mathcal P}$ of~$C$ into parties, and a \emph{distinguished candidate} $p\in C$. 

\noindent {\bf Question:}
Is $p$ a \emph{necessary president}, that is, 
is it true that for all 
possible nominations from parties not containing $p$,  leading to a set~$C'$ of nominees with $p \in C'$, the distinguished candidate~$p$ is a winner of the reduced election~${\E}_{C'}$ over~$C'$?
\end{framed}

\end{minipage} 
\medskip

Notice that we consider the \emph{non-unique winner model}, so we define $p$ to be a necessary president if it is {\it  among the winners} in all possible reduced elections that contain~$p$. 
While the party containing the distinguished candidate~$p$ may contain additional candidates, those are irrelevant in the context of \NPP{}.

\begin{example}
Consider an election with three parties, $A=\{a_1,a_2\}$, $B=\{b_1,b_2\}$, and~$P=\{p\}$.
Assume that each party has to nominate exactly one candidate for an upcoming election that uses the Borda voting rule.
Let the set of voters be $\{v_1,v_2,v_3\}$ with preferences as follows.
\[
\begin{tabular}{llllll}
$v_1:$ 
	&  $p,$ 
	&  $a_1,$ 
	&  $b_1,$ 
	&  $a_2,$ 
	&  $b_2$ 
\\
$v_2:$ 
	&  $a_1,$ 
	&  $p,$ 
	&  $b_1,$ 
	&  $a_2,$ 
	&  $b_2$ 
\\
$v_3:$ 
	&  $b_1,$ 
	&  $b_2,$
	&  $a_2,$  
	&  $p,$ 
	&  $a_1$ 
\end{tabular}
\]
Note that $p$ obtains 2 points from voter~$v_1$ and 2 additional points from~$v_2$ and~$v_3$, irrespective of whether party~$A$ nominates~$a_1$ or~$a_2$: in the former case, $p$ obtains $1$ point from each of these voters, while in the latter case it obtains~$2$ points from~$v_2$  but none from~$v_3$. Hence, $p$'s score in the reduced election will be~$4$ regardless of the nominations from other parties.  
Then it is clear that $p$ is a necessary president, because no other candidate is able to obtain more than~$4$~points, since that would require some candidate of party~$A$ or~$B$ to obtain 2 points from two different voters (namely, from both~$v_2$ and~$v_3$) which is not possible.
\hfill$\lrcorner$
\end{example} 

In order to verify that a given candidate~$p$ is \emph{not} a necessary president, it suffices to present a reduced election containing~$p$ in which $p$ is not a winner. Thus, we have the following fact.
\begin{observation}
\label{obs:incoNP}
\NPP{} is in $\coNP$ for each voting rule where winner determination can be done in polynomial time.    
\end{observation}

As observed for \textsc{Possible President} by, e.g., Schlotter et al.~\citep{schlotter2024}, 
there are at most $s^t$ possible nominations by the parties where $s$ and $t$ are the maximum size and the number of parties, respectively.
Thus, a simple brute force approach yields the following:

\begin{observation}
\label{obs:bruteforce}
\NPP{} is in $\XP$ when parameterized by the number $t$ of parties and in~$\FPT$ when parameterized by both $t$ and the maximum size~$s$ of a party for each voting rule where winner determination can be done in polynomial time.    
\end{observation}

\ifshort
We assume familiarity with the framework of parameterized complexity; see 
the books~\cite{cygan2015parameterized,downey-fellows-FPC-book} %
for an introduction.
\fi 

\subsection{Voting Rules}
In this paper we shall deal with two classes of voting rules, {\it positional scoring rules} and {\it Condorcet-consistent rules}. 
For all considered voting rules the winners can be computed efficiently (that is, in polynomial time) for any election, so by Observation~\ref{obs:incoNP} we know that \NPP{} is in $\coNP$ for all the voting rules studied in this paper.\footnote{We remark that in the case of Ranked Pairs, efficient winner determination assumes some tie-breaking method.}

\subsubsection{Positional Scoring Rules}
 A positional scoring rule for elections involving $t$ candidates is associated with a \emph{scoring vector} $(a_1,a_2,\dots,a_t)$ where $a_1\ge a_2\ge \cdots \ge a_t$ and at least one inequality is strict. For each candidate~$c$, the rule assigns $a_i$ points to~$c$ for each voter that ranks $c$ on the $i^{\mathrm{th}}$ position of her preference list. The \emph{winners} of the election are the candidates with the highest \emph{score}, that is, the total number of points obtained.  We write $\scr(c)$ for the score of  candidate~$c$ in an election~$\E$.

We deal with the following (classes of) positional scoring rules.

 {\it  Short scoring rules}, introduced by Schlotter et al. in~\cite{schlotter2024}, are defined by scoring vectors with only a constant number of non-zero positions, i.e., having the form  $(a_1,a_2,\dots,a_\ell, 0, \dots, 0)$ for some constant~$\ell$. In other words, voters in such elections only allocate points to their $\ell$ most preferred candidates. This class of voting rules contains the well-known scoring rule $\ell$-\emph{Approval} for fixed~$\ell$,  corresponding to the scoring vector with ones in their first $\ell$ positions and zeros afterwards. The case $\ell=1$ is \emph{Plurality} where voters only allocate a single point to their most preferred candidate. 

\emph{Veto-like scoring rules} have scoring vectors that contain some value~$a$ on every position except for the last $\ell$ positions  for some constant~$\ell$, 
i.e., they have the form 
 $(a,\dots,a,a_1,a_2,\dots,a_\ell)$ for some constant $\ell \geq 1$ and $a>a_1$. In other words, voters in such elections distinguish only their $\ell$ least favored candidates.
Veto-like scoring rules include \emph{$\ell$-Veto}, whose scoring vector is $(1,1,\dots,1,0,0,\dots,0)$ with exactly $\ell$ zeros; the case $\ell=1$ is called \emph{Veto}.

Finally, we also consider the \emph{Borda voting rule} which is described by the scoring vector $(t-1,t-2,\dots,1,0)$. Thus, the number of points that a candidate~$c$ receives from a voter~$v \in V$ is the number of candidates ranked worse than~$c$ in the preference list of~$v$.

\subsubsection{Condorcet-Consistent Rules}
 For two candidates $c,c'\in C$, let
  $N_\E(c,c')$ denote the number of voters who prefer candidate~$c$ to candidate~$c'$ in election~$\E$; we shall omit the subscript when $\E$ is clear from the context. If $N_\E(c,c')>N_\E(c',c)$ we say that candidate~$c$ {\it defeats} candidate~$c'$ in~$\E$; if  $N_\E(c,c')=N_\E(c',c)$ and $c \ne c'$, then candidates~$c$ and~$c'$ are {\it tied} in $\E$. 
 The \emph{Condorcet winner} %
 is a candidate that defeats all other candidates; %
 a voting rule is \emph{Condorcet-consistent} if it always selects the Condorcet winner whenever it exists.

The \emph{Copeland$^\alpha$ voting rule} was defined by Faliszewski et al.~\cite{faliszewski2009llull} for some constant $\alpha\in[0,1]$. This voting rule takes all pairs $(c,c')$ of distinct candidates. Considering their head-to-head comparisons, it allocates $1$ point to the candidate defeating the other and allocates $0$ points to the defeated one; being tied earns $\alpha$ points to both candidates.
Formally, %
the score received by~$c$ on the basis of the head-to-head comparison of~$c$ with~$c'$ in $\E$ is
\begin{equation}
\label{eq:define_Copeland}    
\Cpl(c,c')=\left\{ 
\begin{array}{ll}
1  & \text{if $c$ defeats $c'$ in $\E$;} \\
\alpha  & \text{if $c$ and $c'$ are tied in $\E$;} \\
0  & \text{if $c$ is defeated by $c'$ in $\E$.} \\
\end{array}
\right.
\end{equation} 
Then the Copeland$^\alpha$ score of candidate $c$  is computed as the sum
$\Cpl(c)=\sum_{c'\in C\setminus\{c\}}\Cpl(c,c')$.
The winners of~$\E$ are all candidates with the maximum score. Copeland$^\alpha$ for $\alpha=1$ is called the \emph{Llull} rule, and 
we refer to the case $\alpha=0$ as the \emph{Copeland} rule.\footnote{Notice that some papers, in particular \cite{faliszewski2009llull}, use the term Copeland rule for Copeland$^{0.5}$.}

In the \emph{Maximin voting rule}, the Maximin score of candidate~$c$ in election~$\E$ 
is $\MM_\E(c)=\min_{c'\in C \setminus \{ c\}} N_\E(c,c')$.
In other words, the Maximin score of a candidate~$c$ is the largest integer~$r$ such that for every other candidate~$c'$, there exist $r$ voters who prefer~$c$ to~$c'$.
Again, the winners  of~$\E$ are the candidates with maximum score.

The \emph{Ranked Pairs voting rule} uses the so-called majority graph of the election: a directed graph~$D_\E=(C,A)$ where vertices are candidates and  $(c,c')$ is an arc for two distinct candidates~$c$ and~$c'$ if and only if $c$ defeats $c'$ in pairwise comparison in the election~$\E$. 
The winner determination process for Ranked Pairs builds an acyclic\footnote{The arc set~$F$ is \emph{acyclic} if there is no directed cycle in the spanned digraph $(C,F)$; see the textbook by Diestel~\cite{diestel-book} for basic graph terminology.} arc set by considering the set of candidate pairs in~$A$, examined in non-increasing order of their \emph{weight}, %
where the weight of $(c,c') \in A$ is $N_\E(c,c')$.
Starting from an empty arc set~$F$, this process checks whether the currently examined arc can be added to~$F$ without creating any cycles, and if so, adds it to~$F$; the process then proceeds with the next arc. For simplicity, we only consider arcs present in the majority graph~$D_\E$ and, hence, no arcs between tied candidates. Nonetheless, this process usually necessitates some tie-breaking which determines the ordering of arcs with the same weight. The winners are all candidates with no incoming arcs in~$F$.
We remark that our results are not dependent on any particular tie-breaking method, but hold for arbitrary tie-breaking methods.

\subsection{Parameterized Complexity}

A parameterized problem~$Q$ associates with each input instance~$I$ an integer parameter~$k$ that allows us to measure the running time of an algorithm solving~$Q$ as a function of not only the input length~$|I|$ but also the parameter.
We say that $Q$ is \emph{fixed-parameter tractable} (FPT) with parameter~$k$ if it admits an algorithm with running time~$f(k)|I|^{O(1)}$ for some computable function~$f$;
such an algorithm is called an \emph{FPT algorithm} with parameter~$k$.

By contrast, an algorithm that has running time $|I|^{f(k)}$ for some computable function~$f$
is called an \emph{$\mathsf{XP}$ algorithm} with parameter~$k$; parameterized problems admitting such an algorithm are said to be contained in the complexity class $\XP$.
Showing---via a parameterized reduction---that a parameterized problem~$Q$ is $\mathsf{W}[1]$-hard (or even $\mathsf{W}[2]$-hard) implies that $Q$ is not fixed-parameter tractable with the given parameter, unless the standard complexity-theoretic assumption  $\mathsf{FPT} \subsetneq W[1] \subsetneq W[2]$ fails. An even stronger indication of intractability is when $Q$ is $\NP$-hard already for some constant value of the parameter; in this case, $Q$ is called \emph{para-$\NP$-hard}. 

For more background on parameterized complexity, see
the books~\cite{cygan2015parameterized,downey-fellows-FPC-book}.

\section{Results for Positional Scoring Rules}

In this section we present our results on the parameterized complexity of the \NPP{} problem for three types of positional scoring rules. In Section~\ref{sec:borda} we deal with the Borda rule, while in Section~\ref{sec:short+veto} we study short and Veto-like rules.
\subsection{The Borda Rule}
\label{sec:borda}

We start by showing that \NPP{} for the Borda voting rule is polynomial-time solvable. This tractability result is somewhat surprising in view of the fact that the closely related \textsc{Possible President} problem is computationally hard for Borda~\cite{schlotter2024}.
\begin{theorem}
\label{thm:Borda} {\sc Necessary  President} for Borda is polynomial-time solvable.
\end{theorem}

\begin{proof}
We propose  a polynomial-time algorithm that solves \NPP\ for Borda voting; see Algorithm~\ref{alg:Borda} for a pseudocode. Let $I=(\E,\P,p)$ be our input instance. 

Assume that there exists a set~$C' \ni p$ of nominated candidates for which 
$p$ is \emph{not} a winner in the reduced election~$\E_{C'}$ over~$C'$. 
Let $P$ be the party containing~$p$.
Algorithm~\ref{alg:Borda} first guesses
some candidate~$\opp \in \oppParty \in \P \setminus \{P\}$ 
whose score in~$\E_{C'}$ exceeds the score of~$p$, and then greedily nominates a candidate from every remaining party. Note that by ``guessing''%
~$\opp$ we mean iterating over all possibilities for choosing it. 

Algorithm~\ref{alg:Borda} proceeds by computing the following value for each candidate~$c \in C \setminus (P \cup \oppParty )$:
\begin{equation}
\label{def:delta-c}
    \Delta(c)=|\{v:v \in V, \opp \succ_v c \succ_v p\}|-|\{v:v \in V, p \succ_v c \succ_v \opp\}|.
\end{equation}
Notice that %
the quantity $\Delta(c)$ captures the excess of the score of~$\opp$ over the score of~$p$ that results from nominating candidate~$c$.%

Next, for each party~$\widetilde{P} \in \P \setminus \{P,\oppParty \}$, Algorithm~\ref{alg:Borda} nominates a candidate~$c_{\widetilde{P}} \in \widetilde{P}$  maximizing~$\Delta(c_{\widetilde{P}})$ and checks whether %
$p$ is not a winner in the resulting election; if so, it outputs ``no.'' 
If the algorithm has explored all possible guesses for~$\opp$ but has not returned ``no'', then it returns ``yes.''

\begin{varalgorithm}{NP-Borda}
\caption{Solving \NPP{} for Borda.}
\label{alg:Borda}
\begin{algorithmic}[1]
\Require{An instance~$(\E,\P,p)$ of \NPP{} with candidate set~$C$ and $p \in P \in \P$.}
\ForAll{%
$\opp \in C \setminus P$}
    \State Let $\oppParty $ be the party in~$\P$ containing~$\opp$.
    \ForAll{$c' \in C \setminus (P \cup \oppParty )$}\label{line:pickpp}
	   \State Compute $\Delta(c)$ as in (\ref{def:delta-c}).
	\EndFor
    \State Set $\widetilde{C}=\{p,\opp\}$.\label{line:initC}
    \ForAll{$\widetilde{P} \in \P \setminus \{P,\oppParty \}$}
        \State Add some candidate $c_{\widetilde{P}}\in \arg\max_{c \in \widetilde{P}} \Delta(c)$ to~$\widetilde{C}$. \label{line:addtoC}
    \EndFor
    \If{$p$ is not a winner in~$\E_{\widetilde{C}}$}
     {\bf return} ``no''. \label{line:returnyes}
     \EndIf
\EndFor
\State {\bf return} ``yes''. \end{algorithmic}
\end{varalgorithm}

Let us prove the correctness of Algorithm~\ref{alg:Borda}.
It is clear that whenever Algorithm~\ref{alg:Borda} returns ``no'', then it does so correctly, because for a set~$\widetilde{C}$ of nominated candidates (containing exactly one candidate from each party in $\P$), candidate~$p$ nominated by party~$P$ is \emph{not} a winner in the reduced election~$\E_{\widetilde{C}}$ over~$\widetilde{C}$. 

Hence, it remains to prove that whenever the input is a ``no''-instance of \NPP, Algorithm~\ref{alg:Borda} returns ``no.'' Let $C'$ be the set of nominees in some reduced election~$\E_{C'}$ where $p$ is a nominee but not a winner. 
Then there exists some candidate~
$\opp\in \oppParty \in \P \setminus \{P\}$
with $\scr[\E_{C'}](\opp)>\scr[\E_{C'}](p)$.
Let  us denote by $V^{p \succ \opp}$ the set of those voters who prefer $p$ to~$\opp$, and let  $V^{\opp \succ p}=V \setminus V^{p \succ \opp}$ denote the rest of the voters.

Our key observation is the following: if for some voter~$v \in V^{\opp \succ p}$ there are exactly $i$ nominees in~$\E_{C'}$ that $v$ prefers to~$p$ but not to~$\opp$,  then the Borda rule allocates $i+1$ points more to~$\opp$ than to~$p$ due to voter~$v$ in~$\E_{C'}$. Similarly, if for some voter~$v \in V^{p \succ \opp}$ there are exactly~$i$ nominees in~$\E_{C'}$ that $v$ prefers to~$\opp$ but not to~$p$,  then Borda allocates $i+1$ points more to~$p$ than to~$\opp$ due to voter~$v$. 
Therefore, we get that 
\begin{align*}
    0 & <  \scr[\E_{C'}](\opp) -\scr[\E_{C'}](p) = \\
    &= \sum_{v \in V^{\opp \succ p}} 
    \bigg(\,\bigg| \! \! \bigcup_{\substack{c \in C',\\ \opp \succ_v c \succ_v p}} \!\!\! \{c\}\bigg|+1\bigg)  -
    \sum_{v \in V^{p \succ \opp}} 
    \bigg(\,\bigg| \! \! \bigcup_{\substack{c \in C',\\ p \succ_v c \succ_v \opp}} \!\!\! \{c\}\bigg|+1\bigg) \\
    &= \sum_{c \in C' \setminus \{p,\opp\}} \! \Delta(c) 
    +|\{v:v \in V,\opp \succ_v p\}|-
    |\{v:v \in V,p \succ_v \opp\}|
    \\
    & \leq \sum_{c \in \widetilde{C} \setminus \{p,\opp\}} \! \Delta(c) 
    +|\{v:v \in V,\opp \succ_v p\}|-
    |\{v:v \in V,p \succ_v \opp\}| \\
    &= \sum_{v \in V^{\opp \succ p}}
    \bigg(\,\bigg| \! \! \bigcup_{\substack{c \in \widetilde{C},\\ \opp \succ_v c \succ_v p}} \!\!\! \{c\}\bigg|+1\bigg)  -
    \sum_{v \in V^{p \succ \opp}}
    \bigg(\,\bigg| \! \! \bigcup_{\substack{c \in \widetilde{C},\\ p \succ_v c \succ_v \opp}} \!\!\! \{c\}\bigg|+1\bigg)
    \\
    &= \scr[\E_{\widetilde{C}}](\opp)-\scr[\E_{\widetilde{C}}](p).
\end{align*}
where $\widetilde{C}$ is the set computed by Algorithm~\ref{alg:Borda} on lines~\ref{line:initC}--\ref{line:addtoC} during the iteration where %
$\opp$ is picked on line~\ref{line:pickpp}.
Note that the inequality follows from the fact that each nominee~$c \in \widetilde{C}$ maximizes~$\Delta(c)$ among all candidates within its own party, as ensured by line~\ref{line:addtoC}. It follows that Algorithm~\ref{alg:Borda} will find on line~\ref{line:returnyes} that $p$ is not a winner of the election~$\E_{\widetilde{C}}$ because $\scr[\E_{\widetilde{C}}](\opp)>\scr[\E_{\widetilde{C}}](p)$, and hence outputs ``no'' as required.

Note that there are less than~$|C|$ possibilities to choose~$\opp$. The values~$\Delta(c)$ can be computed in time $O(|V| \cdot |C|)$ and then $\widetilde{C}$ can also constructed in time $O(|C|)$. Finally, the score of~$p$ and~$\opp$ in the reduced election over~$\widetilde{C}$ can also be computed in  $O(|V| \cdot |C|)$ time, so the total running time of Algorithm~\ref{alg:Borda} is $O(|C|^2 \cdot |V|)$.
\end{proof}

\subsection{Short and Veto-Like Scoring Rules}
\label{sec:short+veto}

In this section we examine the computational complexity of the \NPP{} problem in detail for all short and Veto-like scoring rules.
We begin by proving that \NPP{} for such voting rules is $\coNP$-complete even if the maximum party size is $\maxsize=2$.
Thus, for short scoring rules we generalize a result by Faliszewski et al.~\cite{faliszewski2016} who showed that \NPP{} for Plurality is $\coNP$-complete.

The proofs of Theorems~\ref{thm:short-coNP-s} and \ref{thm:veto-coNP-s}
are based on a polynomial reduction from  {\sc(2,2)-E3-SAT},  
the problem of deciding whether a given 3-CNF formula 
where each variable occurs twice as a positive and twice as a negative literal is satisfiable;
this problem was proved to be  $\NP$-complete by Berman et al.~\cite{berman-karpinski-scott-balanced3SAT}.
Before presenting Theorem~\ref{thm:short-coNP-s}, let us 
formally define the problem \textsc{(2,2)-E3-SAT}:

\vskip1ex

\begin{minipage}{\columnwidth}
\begin{framed}
\noindent {\bf Problem }{\sc(2,2)-E3-SAT}

\noindent {\bf Instance:} A CNF Boolean formula $\varphi=C_1 \wedge \dots \wedge C_m$ over a set~$X=\{x_1,\dots,x_n\}$ of variables where each clause~$C_j$ contains exactly three distinct literals and each variable occurs twice as a positive and twice as a negative literal.

\noindent {\bf Question:}
Is $\varphi$ satisfiable?
    
\end{framed}    
\end{minipage}

\vskip1ex

Henceforth, we let $y_j^1$, $y_j^2$, and $y_j^3$ denote the three literals in~$C_j$ for each $j \in [m]$.

\begin{theorem}
\label{thm:short-coNP-s}
Let $\R$ be a short voting rule based on a positional scoring vector that has the form $(a_1,a_2,\dots,a_{\ell},0,\dots,0)$ for some constant~$\ell \geq 1$ such that $a_{\ell}>0$. Then
\NPP{} for~$\R$ is $\coNP$-complete even if the maximum party size is $\maxsize=2$.
\end{theorem}

\begin{proof}
Let $\varphi=C_1 \wedge \dots \wedge C_m$ over variable set~$X=\{x_1,\dots,x_n\}$ be our instance of {\sc(2,2)-E3-SAT}. The structure of a {\sc(2,2)-E3-SAT} Boolean formula implies that  $m \geq 2$.

We shall construct an instance $I$ of \NPP\ with distinguished candidate~$p$ contained in party~$P=\{p\}$ in a way such that $p$ is the unique winner of the reduced election for every possible set of nominations by other parties if and only if $\varphi$ is \emph{not} satisfiable.

In addition to~$P$, let us define party $\oppParty  = \{\opp\}$ and a party~$X_i=\{x_i, \overline{x}_i\}$ for each variable~$x_i \in X$. Moreover, we additionally define a set~$D_1\cup D_2\cup D_3\cup D_4$ of dummy candidates with $|D_h|=\ell-1$ for $h=1,2,3,4$, with each dummy forming its own single-candidate party. Observe that the number of parties is $n+4\ell-2$ and the maximum size of any party is $\maxsize=2$.

The set of voters is $V = \{w\}\cup U\cup U' \cup V_{0} \cup V'_{0}$ where $|V_0|=|V_0'|=2m+1$ and $U=\{u_1,\dots,u_m\}, U'=\{u'_1,\dots,u'_m\}$. Thus, there are $6m+3$ voters.
The preference profile is shown below. Note that in this and other proofs $[\cdots]$ denotes the remaining candidates (not explicitly stated in  the preference list) in an arbitrary strict order; moreover, when we write a  subset of candidates in a preference list, this means  the candidates of this set written in an arbitrary strict order. 

\[
\begin{array}{ll}
    w: &  D_1 \succ \opp \succ [\cdots];\\
       v \in V_{0}: &  \opp \succ D_2  \succ [\cdots];\\
    v' \in V'_{0}:  &  p \succ D_3 \succ [\cdots];\\
     u_j,u'_j \textrm{ for } j \in [m]:  & y_j^1 \succ y_j^2 \succ y_j^3 \succ D_4 \succ p \succ [\cdots]. \\
\end{array}
\]

Consider a reduced election~$\E$.
We see that candidate~$\opp$ earns $a_1$ points from each voter in~$V_0$, $a_{\ell}$ points from voter~$w$ and, thanks to dummy candidates, no points elsewhere. Hence, we get  
$\scr(\opp)=(2m+1)a_1 + a_{\ell}$.
Candidate~$p$ receives $(2m+1)a_1$ points from the voters in~$V'_0$. Since each literal occurs at most twice in~$\varphi$, both candidates in~$X_i$ receive at most $4a_1$ points from voters in $U\cup U'$. Dummy candidates (present only if $\ell\geq 2$) from $D_1, D_2,D_3,$ and~$D_4$ receive at most $a_1$, $2ma_1$, $(2m+1)a_2$, and  $(2m+1)a_2$ points, respectively. Each of these values is less than  $\scr(\opp)$.

Now assume that there is a satisfying truth assignment for~$\varphi$. Let each party~$X_i$ nominate the candidate corresponding to the true literal in~$\{x_i,\overline{x}_i\}$. Since each clause contains at least one true literal, we know that at least one candidate in~$\{y_j^1,y_j^2,y_j^3\}$ is nominated for each $j \in [m]$. 
Therefore, candidate~$p$ receives no additional points from voters in~$U \cup U'$ and  she is not  the  winner in the resulting election as $\scr(p)=(2m+1)a_1<\scr(\opp)$.

Conversely,  assume that $\varphi$ admits no satisfying truth assignment. Hence, 
irrespective of the nominations from parties~$X_i$, $i \in [n]$, there exists at least one clause~$C_j$ for which no literal in~$\{y_j^1,y_j^2,y_j^3\}$ is set to true, implying that both voters $u_j$ and $u'_j$ allocate $a_\ell$ additional points to candidate~$p$. This ensures that $p$ is the unique winner of any obtained reduced election~$\E$ with  $\scr(p)=(2m+1)a_1+2a_\ell$; hence, she is the necessary president in $I$.
\end{proof}

Let us now turn our attention to Veto-like scoring rules.
The proof of Theorem~\ref{thm:veto-coNP-s} also presents a reduction from the \textsc{(2,2)-E3-SAT} problem; recall the notation after the definition of \textsc{(2,2)-E3-SAT} given
at the beginning of the section.

\begin{theorem}
\label{thm:veto-coNP-s}
Let $\R$ be a Veto-like scoring voting rule, based on a positional scoring vector that has the form $(a,\dots,a,a_1,a_2,\dots,a_{\ell})$ for some constant~$\ell \geq 1$ such that $a> a_{\ell}$. Then
\NPP{} for~$\R$ is $\coNP$-complete even if the maximum party size is $\maxsize=2$.
\end{theorem}

\begin{proof}
We again present a polynomial reduction from  {\sc(2,2)-E3-SAT}.
We shall construct an instance $I$ of \NPP\ with distinguished candidate~$p$ contained in party~$P=\{p\}$ in a way such that $p$ is the  winner of the reduced election for every possible set of nominations by other parties if and only if $\varphi$ is \emph{not} satisfiable.

In addition to~$P$, let us define party $\oppParty  = \{\opp\}$ and a party $X_i=\{x_i, \overline{x}_i\}$ for each variable~$x_i \in X$. Moreover, we create a set~$D$ of~$\ell-1$ dummy candidates, with each dummy constituting its own single-candidate party. Observe that the number of parties is $n+\ell+1$ and the maximum size of any party is $\maxsize=2$.

The set of voters is defined as $V=\{v_0\}\cup W\cup \bigcup_{j \in [m]}V_j$ where $W=\bigcup_{i \in [n]} W_i$, $|W_i|=2$ for each $i\in[n]$ and $|V_j|=2$ for each ${j\in [m]}$. The number of voters is thus $2n+2m+1$ and their  preferences are as follows:
\[
\begin{array}{ll}
    v_0: & [\cdots]  \succ p \succ D;\\
    w \in W_i \text{ for }i \in [n]:  &  [\cdots] \succ x_i \succ \overline{x}_i \succ D;\\
    v \in V_j  \text{ for }j \in [m]:   &  [\cdots]\succ \opp\succ y_j^1 \succ y_j^2 \succ y_j^3 \succ D.\\
        \end{array}
\] 

Observe that
$\scr(p)=(2n+2m)a+a_1$
and since each nominee of party $X_i$ receives from the voters in $W_i$ only $a_1$ points, we have that
$\scr(x_i)$ as well $\scr(\overline{x}_i)$ is at most $(2n+2m-1)a+2a_1$, which is strictly smaller than $\scr(p)$ and so a nominee of any party $X_i$ can never be the winner of any reduced election. Further, candidate~$\opp$ receives $(2n+1)a$ points from voters in $\{v_0\}\cup W$. 

Now assume that there is a satisfying truth assignment for~$\varphi$. Let each party~$X_i$ nominate the candidate corresponding to the true literal in~$\{x_i,\overline{x}_i\}$. Since each clause contains at least one true literal, at least one candidate in~$\{y_j^1,y_j^2,y_j^3\}$ is nominated for each $j \in [m]$. 
Hence, $\scr(\opp)=(2n+2m+1)a>\scr(p)$, so for these nominations candidate~$\opp$ is the winner of  the resulting reduced election.

Conversely,  assume that $\varphi$ admits no satisfying truth assignment. Hence, irrespective of the nominations from parties~$X_i$, $i \in [n]$, there exists at least one clause~$C_j$ for which no literal in~$\{y_j^1,y_j^2,y_j^3\}$ is set to true. Thus, candidate~$\opp$ receives less than $a$ points from voters in $V_j$, so $\scr(\opp)\le (2n+2m-1)a+2a_1<\scr(p)$. Therefore $p$ is the unique winner of any reduced election, so she is the necessary president in $I$.
\end{proof}

Given the intractability results in Theorems~\ref{thm:short-coNP-s} and~\ref{thm:veto-coNP-s}, we investigate the parameterized complexity of \NPP{} for short and Veto-like scoring rules. In particular, we consider the number~$t$ of parties and the number~$\tau$ of voter types as parameters.

\subsubsection{Parameterizing by the Number of Parties}

Here we show that \NPP{} is $\mathsf{W}[2]$-hard for both short and Veto-like scoring rules with parameter~$t$. Our proofs use similar ideas as the proofs of Theorems~1 and~6 in \cite{schlotter2024}, respectively. For both classes of scoring rules, 
we provide a parameterized reduction from the classic \textsc{Hitting Set} problem which is $\mathsf{W}[2]$-hard when parameterized by the size of the desired hitting set (see Cygan et al.~\cite{cygan2015parameterized}).
Let us formally define this problem:

\begin{framed}
\noindent {\bf Problem }{\sc Hitting Set}

\noindent {\bf Instance:} A tuple $H=(S, {\mathcal F},k)$ where $S = \{s_1, s_2, \dots, s_n\}$ is a set of elements, $\mathcal{F} = \{F_1, F_2, \dots, F_m\} \subseteq 2^{S}$ a family of subsets of $S$, and $k$ an integer.

\noindent {\bf Question:}
Does there exist a set $S'\subseteq S$ such that $|S'|\le k$ and $S'\cap F_j\ne \emptyset$ for each $j\in [m]$?
    
\end{framed}

\begin{theorem}
\label{thm:short-W2hard-t}
Let $\R$ be a short voting rule, based on a positional scoring vector %
of the form $(a_1,a_2,\dots,a_{\ell},0,\dots,0)$ for some constant~$\ell \geq 1$ such that $a_{\ell}>0$. Then
\NPP{} for~$\R$ is  
$\mathsf{W}[2]$-hard  when parameterized by $t$, the number of parties.
\end{theorem}

\begin{proof}
We present a parameterized reduction from the {\sc Hitting Set} problem.

 For any instance $H=(S, {\mathcal F},k)$ of the {\sc Hitting Set}  we shall construct an instance $I$  of \NPP{} with distinguished candidate~$p$ contained in party $P = \{p\}$ in such a way that $p$ is a winner of the reduced election for every possible nominations by other parties if and only if  $H$ is ``no''-instance of \textsc{Hitting Set}.

In addition to~$P$, let us define parties $\oppParty  = \{\opp\}$ and $P_i = \{s^i_1, \dots, s^i_n\}$ for each~$i \in [k]$, where candidate~$s^i_r$ represents the $i^\textrm{th}$ copy of the element $s_r\in S$ for some $r\in [n]$. We also use the notation  $F_j^i=\{s^i : s\in F_j\}$ for the set of the $i^\textrm{th}$ copies of the elements contained in $F_j$, $j\in [m]$. Moreover, we additionally define a set~$D_1\cup D_2\cup D_3\cup D_4$ of dummy candidates with $|D_h|=\ell-1$ for $h=1,2,3,4$,  with each dummy forming its own single-candidate party. Therefore, there are altogether $t=k+4\ell-2$ parties. 

The set of voters is $V = \{w\}\cup U\cup U' \cup V_{0} \cup V'_{0}$ where $|V_0|=|V_0'|=2m+1$ and $U=\{u_1,\dots,u_m\}, U'=\{u'_1,\dots,u'_m\}$. Thus, there are $6m+3$ voters.
The preference profile is shown below.

\[
\begin{array}{ll}
     w: &  D_1 \succ \opp \succ [\cdots];\\
    u_j,u'_j \textrm{ for } j \in [m]:  &  F_j^1 \succ F_j^2 \succ \dots \succ F_j^k \succ D_2 \succ p \succ [\cdots]; \\
    v \in V_{0}: &  \opp \succ D_3  \succ [\cdots];\\
    v' \in V'_{0}:  &  p \succ D_4 \succ [\cdots].\\
\end{array}
\]

Recall that $\ell$ is a constant, so the number of parties is only a function of~$k$, and thus the presented reduction is a parameterized reduction.
It is also a polynomial-time reduction, so by the $\NP$-hardness of \textsc{Hitting Set},  our proof yields not only $W[2]$-hardness for parameter~$t$, but also $\mathsf{NP}$-hardness.

Consider a reduced election~$\E$.
We see that candidate~$\opp$ earns $a_1$ points from each voter in $V_0$, $a_{\ell}$ points from voter $w$ and, thanks to dummy candidates, no points elsewhere. Hence, we get 
$\scr(\opp)=(2m+1)a_1 + a_{\ell}$.
For candidate~$p$ we have $\scr(p)=(2m+1)a_1$, she receives all her points  from the voters in $V'_0$. Each candidate~$s_r^i$ receives at most $2ma_1$ points from voters in $U\cup U'$. Dummy candidates (present only if $\ell\geq 2$) from $D_1, D_2,D_3,$ and $D_4$ receive at most $a_1$, $2ma_1$, $(2m+1)a_2$, and \hbox{$(2m+1)a_2$} points, respectively. Each of these values is less than  $\scr(\opp)$.

Now assume that there is a hitting set $S' = \{s_{(1)}, \dots, s_{(k)}\}$ for the instance $H$ where $s_{(i)}$ denotes the $i^\textrm{th}$ element in $S'$ for some fixed order. Let each party $P_i$ nominate the  candidate corresponding to the $i^\textrm{th}$ copy of  element $s_{(i)} \in S'$, i.e., the candidate $s_{(i)}^i$. Since $S' \cap F_j \neq \emptyset$ for each $j \in [m]$, this ensures that at least one candidate in the set $F_j^1 \cup F_j^2 \dots \cup F_j^k$ is nominated. 
Therefore, candidate~$p$ receives no additional points from voters in~$U \cup U'$, and  she is not among the winners of the resulting reduced  election.

Conversely,  assume that $H$ does not admit any hitting set of size~$k$. Hence, 
for each subset $S'\subseteq S$ of size at most~$k$  there exists at least one set~$F_j\in {\mathcal F}$ with $S'\cap F_j=\emptyset$. This in turn implies  that for any nominations by parties $P_i$, $i \in [k]$, there exists at least one index~$j$ such that no candidate in~$F_j^1 \cup F_j^2 \cup \dots \cup F_j^k$ is nominated and, therefore, both voters $u_j$ and $u'_j$ allocate $a_\ell$ additional points to candidate~$p$, ensuring that $p$ is the unique winner of the obtained reduced election~$\E$ with $\scr(p)=(2m+1)a_1+2a_\ell$.

Hence, $p$ is the winner of any reduced election, thus the necessary  winner in $I$, if and only if $H$ is a no-instance of the {\sc Hitting Set}. 
\end{proof}

Next, we present the analog of Theorem~\ref{thm:short-W2hard-t} for Veto-like voting rules.

\begin{theorem}
\label{thm:Veto-W2hard-t}
Let $\R$ be a Veto-like voting rule, based on a scoring vector of the form  $(a,\dots,a,a_1,a_2,\dots,a_\ell)$ for some constant $\ell \geq 1$ such that $a>a_1$. Then 
  \NPP{}  for~$\R$ is %
  $\mathsf{W}[2]$-hard  when parameterized by $t$, the number of parties.
\end{theorem}

\begin{proof}
We  again present a parameterized reduction from {\sc Hitting Set}. Let $H$ be an instance of {\sc Hitting Set}, we construct and instance $I$ of \NPP.

Again, $p$ is our distinguished candidate contained in the singleton party ${P = \{p\}}$, and we define the parties $\oppParty =\{\opp\}$ and  $P_i = \{s^i_1, \dots, s^i_n\}$ for each~$i \in [k]$. For each $j \in [m]$, we write $F_j^i=\{s^i : s\in F_j\}$ for the set of $i^\textrm{th}$ copies of the elements contained in $F_j$.
Further, we create a set $D$ of $\ell-1$ dummy candidates, with each dummy constituting its own single-candidate party. The number of parties is therefore $t=k+\ell+1$. The presented reduction is a parameterized reduction, since $\ell$ is a constant, and so the number of parties is only a function of $k$.

The set of voters is  
\[V=\{v_0\}\cup W\cup V'
\]
where W=$\bigcup_{i \in [m]}W_i$, $V'=\bigcup_{j \in [k]} V_j$, $|W_i|=2$ for each $i\in[k]$ and $|V_j|=2$ for each $j\in [m]$. The number of voters is thus $2k+2m+1$. Their  preferences are as follows:
\[
\begin{array}{ll}
    v_0: & [\cdots]  \succ p \succ D;\\
    w \in W_i \text{ for }i \in [k]:  &  [\cdots] \succ s_1^i \succ s_2^i \succ \dots \succ s_n^i \succ D;\\
    v \in V_j \text{ for }j \in [m]:   &  [\cdots]\succ \opp\succ F_j^1 \succ F_j^2 \succ \dots \succ F_j^k \succ D.
    
    \end{array}
\] 
It is easy to see that for any reduced  election $\E$ we have 
\begin{eqnarray*}
    \scr(p)&=&(2k+2m)a+a_1;\\
    \scr(d)&\le& (2k+2m+1)a_2   \mbox{\ for any\ }d\in D;\\
    \scr(s^i_r)&\le& (2m+1)a+2ka_1  \mbox{\ for any\ }s^i_r; i\in[k], r\in[n],
\end{eqnarray*}
which implies $\scr(d)<scr(p)$ for each dummy candidate ${d\in D}$ and $\scr(p_i)<scr(p)$ for each nominee of party $P_i$, $i\in [k]$.
Candidate $\opp$ receives $(2k+1)a$ points from voters in $\{v_0\}\cup W$.

Now assume that  a hitting set $S'=\{ s_{(1)}, s_{(2)},\dots, s_{(k)}\}\subseteq S$ exists for  instance $H$. Let party $P_i$ nominate its member $s_{(i)}^i$ for each $i\in[k]$. Then candidate $\opp$ is in the preference list of each voter in  $V'$  followed by at least $\ell$ other candidates, so her score becomes $\scr(\opp)=(2k+2m+1)a$ and therefore for these nominations she is the only winner of $\E$. Thus $p$ is not a necessary president.

Conversely, assume that $H$ does not admit a hitting set. Then there exists  at least one index $j\in [m]$ with $S' \cap F_j=\emptyset$. This means that for any possible nominations from parties $P_i$, there exists some $ j \in [m]$ for which the dummy candidates  immediately follow candidate~$\opp$ in the preferences of voters in~$V_j$ and so $\scr(\opp)\le (2k+1)a+(2m-2)a+a_1$. This is strictly smaller than $\scr(p)$ and therefore candidate~$p$ is a necessary president in $I$.
\end{proof}

\subsubsection{Parameterizing by the Number of Voter Types}

Next, we show that \NPP{}  becomes fixed-parameter tractable for both short and Veto-like positional scoring rules when parameterized by the number~$\tau$ of voter types. 
This contrasts sharply our intractability results for parameterizing by the number or maximum size of parties, as presented in Theorems~\ref{thm:short-coNP-s}--\ref{thm:Veto-W2hard-t}.

\begin{theorem}\label{thm:FPT-short}
    Let $\R$ be a short voting rule, based on a scoring vector of the form     $(a_1,a_2,\dots,a_\ell,0,\dots,0)$ for some $\ell \geq 1$ such that $a_\ell>0$. Then
    \NPP{} for~$\R$ is $\FPT$ when parameterized by~$\tau$, the number of voter types.
\end{theorem}
\begin{proof}
    We present Algorithm~\ref{alg:short-FPT} to solve an instance $I=(\E,\mathcal{P},p)$
    of \NPP\ for~$\R$ in $\FPT$ time with parameter~$\tau$, the number of different voter types in~$\E$.  Let $V=V_1 \cup \dots \cup V_\tau$ be the partitioning of the voters by their types, i.e., all voters in~$V_i$ for some $i \in [\tau]$ have the same preferences over~$C$. 

    Assume that there exists a set~$C' \ni p$ of nominated candidates for which %
    $p$ is \emph{not} a winner in the reduced election~$\E_{C'}$ over~$C'$. 
    Let $P$ be the party containing~$p$.
    Algorithm~\ref{alg:short-FPT} guesses the following information about~$\E_{C'}$:
    \begin{itemize}
        \item 
        A candidate~$\opp \in C'$ whose score in~$\E_{C'}$ exceeds that of~$p$.
        There are $|C \setminus P|$ possibilities to choose~$\opp$.
        \item The \emph{structure} of~$\E_{C'}$, defined as follows.
        We first define a function $\place:[\tau] \times [\ell] \to \P$ as follows.
    For some voter type $i \in [\tau]$ and index $j \in [\ell]$, 
    let $\place(i,j)$ denote the party containing the candidate at the $j^{\textup{th}}$ position of the votes from~$V_i$ in the reduced election~$\E_{C'}$.
    For some pairs  
     $(i,j)$ and $(i',j')$ in~$[\tau] \times [\ell]$,   we write $(i,j) \mysim (i',j')$  if $\place(i,j)=\place(i',j')$. 
    The \emph{structure of~$\E_{C'}$} is  the family~$\Q$ of equivalence classes of the relation~$\mysim$, which is a partitioning of $[\tau] \times [\ell]$. 
    Note that there are at most $(\tau \ell)^{\tau \ell}$ possibilities to choose $\Q$.
        \item The equivalence class~$Q_{\opp} \in \Q$ containing all pairs~$(i,j)$ for which $\place(i,j)=\oppParty $. 
        There are $|\Q| \leq \tau \ell$ possibilities to choose $Q_{\opp}$.
        \item The set $Q_{p}$ containing all pairs~$(i,j)$ for which $\place(i,j)=P$. Note that either $Q_p \in \Q$ or  $Q_{p} = \emptyset$. Therefore, there are $|\Q \setminus \{Q_{\opp}\}|+1 \leq \tau \ell$ possibilities to choose $Q_{p}$.
    \end{itemize}

    After guessing the above described information about~$\E_{C'}$ (where by ``guessing'' we mean trying all possibilities), Algorithm~\ref{alg:short-FPT} proceeds by creating an auxiliary bipartite graph~$G$ as follows. The vertex set of~$G$ is  $\hat{\P} \cup \hat{\Q}$ where $\hat{\P}=\P \setminus \{P,\oppParty \}$ and $\hat{\Q}=\Q \setminus \{Q_p,Q_{\opp}\}$.
    For each $Q \in \Q$, let us define \[Q^+=Q \cup \{(i,\ell+1):i \in [\tau],\not \exists j \in [\ell] \text{ such that } (i,j) \in Q\}.\]
    Intuitively $Q^+$ represents the situation of a party whose nominee (i)~obtains the $j^\textup{th}$ position in the votes from voters of~$V_i$ in~$\E_{C'}$ for each $(i,j) \in Q$ with $j \in [\ell]$, and (ii) obtains a position after the $\ell^{\textup{th}}$ position  in the votes from voters of~$V_i$ in~$\E_{C'}$ for each $(i,\ell+1) \in Q^+$.
    
    A candidate~$c \in C \setminus (P \cup \oppParty ) $
    is \emph{well placed (with respect to~$p$ and~$\opp$) for~$Q \in \hat{\Q}$} if for each $i \in [\tau]$ the following hold:
    \begin{itemize}
        \item[(i)] if $(i,j) \in Q^+$ and $(i,j') \in Q^+_{p}$ for some $j$ and~$j'$, then either $j=j'=\ell+1$ or $c$ precedes~$p$ in the preferences of voters from~$V_i$ if and only if $j<j'$;
        \item[(ii)] if $(i,j) \in Q^+$ and $(i,j'') \in Q^+_{\opp}$ for some $j$ and~$j''$, then either $j=j''= \ell+1$ or $c$ precedes~$\opp$ in the preferences of voters from~$V_i$ if and only if $j<j''$.
    \end{itemize}
    
    We can now define the auxiliary graph~$G$: some party~$\widetilde{P} \in \hat{\P}$ is connected by an edge with some $Q \in \hat{\Q}$ in~$G$ if and only if 
    there exists a candidate~$c \in \widetilde{P}$ that is well placed with respect to~$p$ and~$\opp$ for~$Q$. 
    Given candidate~$c$ and some $Q \in \hat{\Q}$, it can be checked in $O(\tau  \ell)$ time whether $c$ is well placed for~$Q$, so the graph~$G$ can be computed in~$O(|C| \tau^2 \ell^2)$ time.

After computing the graph~$G$, Algorithm~\ref{alg:short-FPT} also computes a set~$\S \subseteq \P$ of 
\emph{secure} parties which are those parties~$S$ that contain at least one candidate~$c_S$ that is \emph{safe}, meaning that for each $i \in [\tau]$ where $(i,\ell+1) \notin Q^+_{\opp}$, voters in~$V_i$ prefer~$\opp$ to~$c_S$; we fix one such candidate~$c_S$ for each secure party~$S$. Intuitively, such candidates can be safely nominated in the sense that they will not prevent~$\opp$ from obtaining the required points in~$\E$.

Next, Algorithm~\ref{alg:short-FPT} computes a matching~$M$ in~$G$ covering all non-secure parties in~$\hat{\P} \setminus \S$ and all equivalence classes in $\hat{\Q}$; if no such matching exists, it discards the current set of guesses. 
Each party~$\widetilde{P} \in \hat{\P}$ that is covered by an edge~$(\widetilde{P},Q) \in M$ nominates a candidate that is well placed for~$Q$, whereas each party~$S$ not covered by an edge of~$M$ (which is necessarily a secure party) nominates the candidate~$c_S$. 
Finally, the algorithm checks whether these nominations together with~$p$ and~$\opp$ yield a reduced election~$\E^\star$ in which $p$ is not a winner; if so, it returns ``no.''
If the algorithm has explored all possible guesses but has not returned ``no'', then it returns ``yes.''

\begin{varalgorithm}{NP-Short}
\caption{Solving \NPP\ for short voting rules.}
\label{alg:short-FPT}
\begin{algorithmic}[1]
\Require{An instance~$(\E,\P,p)$ of \NPP{} over candidate set~$C$ and $p \in P \in \P$.}
\ForAll{%
$\opp \in C \setminus P$}\label{line:sh-pickpp}
 \ForAll{partitioning~$\Q$ of~$[\tau] \times [\ell]$} \Comment{$\Q$: structure of~$\E$}
  \ForAll{$Q_{\opp} \in \Q$ and $Q_p \in (\Q \setminus \{Q_{\opp}\}) \cup \{\emptyset\}$}\label{line:sh-pickQpQp} 
    \State Let $\oppParty $ be the party in~$\P$ containing $\opp$.
    \State Let $\hat{\P}=\P \setminus \{P,\oppParty \}$ and $\hat{\Q}=\Q \setminus \{Q_p,Q_{\opp}\}$.
    \State Let $E=\emptyset$ and $\S=\emptyset$.
    \ForAll{$c \in P_c \in \hat{\P}$ and $Q \in \hat{\Q}$}
        \If{$c$ is well placed for~$Q$} add $(P_c,Q)$ to $E$.
        \EndIf
    \EndFor
    \State Create the graph $G=(\hat{\P} \cup \hat{\Q},E)$.
    \ForAll{$S \in \hat{\P}$} 
        \If{$\exists$ a safe candidate $c_S \in S$} add $S$ to~$\S$. 
        \EndIf
    \EndFor
    \If{$\exists$ a matching~$M$ in~$G$ covering~$\hat{\Q} \cup  (\hat{\P} \setminus \S)$}\label{line:sh-computematching}
        \State Set $C_M=\{p,\opp\}$.\label{line:sh-initC} 
        \ForAll{$(\widetilde{P},Q) \in M$} 
            \State add to~$C_M$ a candidate of~$\widetilde{P}$ well placed for~$Q$.
        \EndFor
        \ForAll{$S \in \S$ not covered by~$M$} add $c_S$ to~$C_M$.\label{line:sh-finishC} 
        \EndFor
        \If{$p$ is not a winner in $\E_{C_M}$}  {\bf return} ``no''. \label{line:sh-returnyes}
        \EndIf
    \EndIf
  \EndFor
 \EndFor
\EndFor
\State {\bf return} ``yes''. \end{algorithmic}
\end{varalgorithm}

Note that there are at most $%
|C| \cdot (\tau \ell)^{\tau \ell+2}$ possibilities for picking~$%
\opp,\Q,Q_p,$ and~$Q_{\opp}$. 
Once these guesses are fixed, the bipartite auxiliary graph~$G$ can be computed in~$O(|C| \tau^2  \ell^2)$ time as we have already argued. The bottleneck in each iteration is the computation of the matching on line~\ref{line:sh-computematching} of Algorithm~\ref{alg:short-FPT}. Since $G$ has at most $|\P|+\tau \ell\leq |C|+\tau\ell$ vertices, we can compute a  matching in~$G$ that covers the required set of vertices in $O((|C|+\tau \ell)^{2.5})$ time by, e.g., the Hopcroft--Karp algorithm~\cite{HopcroftKarp73}. %
This yields a total running time of $O(%
|C|^{3.5} \cdot (\tau \ell)^{\tau \ell+4.5})$ which is fixed-parameter tractable with respect to~$\tau$, as $\ell$ is a constant.

\medskip
Let us show the correctness of Algorithm~\ref{alg:short-FPT}. It is clear that whenever Algorithm~\ref{alg:short-FPT} returns ``no'', then it does so correctly, because for a set~$C_M$ of nominated candidates containing~$p$ as well as exactly one candidate for each party in~$\P\setminus \{P\}$, candidate~$p$ is not a winner of the reduced election over~$C_M$ and thus not a necessary president.

Hence, it remains to prove that whenever the input is a ``no''-instance of \NPP, Algorithm~\ref{alg:short-FPT} returns ``no.'' Let $C'$ be the set of nominees in the reduced election~$\E_{C'}$ where $P$ is not a winner. 
Then there exists  some candidate~$\opp$ with $\scr(\opp)>\scr(p)$. %
Let also $\oppParty  \in \P$ contain $\opp$. 
Let $\Q$ be the structure of~$\E$, and let $Q_{\opp} \in \Q$ and $Q_p \in \Q \cup \{\emptyset\}$ be defined as before. 
Consider the iteration~$\iota$ when the algorithm chooses candidate~$\opp$, the structure~$\Q$, and the sets~$Q_p$ and~$Q_{\opp}$ on lines~\ref{line:sh-pickpp}--\ref{line:sh-pickQpQp}. 

Observe that for each $(i,j) \in Q$ for some $Q  \in \hat{\Q}$, the candidate~$c$ that appears on the $j^{\textup{th}}$ position of the votes from~$V_i$ in~$\E_{C'}$ for some $j \in [\ell]$ is well placed for~$Q$ (due to the definition of the structure~$\Q$ of~$\E_{C'}$), and thus the party nominating~$c$ is connected to~$Q$ by an edge in the auxiliary graph~$G$. 
Let $M^\star$ denote the set of all such edges; then $M^\star$ is a matching in~$G$ that covers every equivalence class in~$\hat{\Q}$. 
Furthermore, notice that every party that receives a total score of~$0$ in~$\E$ has a nominee in~$\E_{C'}$ that is safe. 
Therefore, the matching~$M^\star$ covers all non-secure parties. Consequently, Algorithm~\ref{alg:short-FPT} will find on line~\ref{line:sh-computematching} that there exists \emph{some} matching~$M$ covering $\hat{\Q}  \cup (\hat{\P} \setminus \S)$.

We now show that $p$ cannot be a winner in the reduced election~$\E_{C_M}$ over the set~$C_M$ of candidates computed by Algorithm~\ref{alg:short-FPT} on lines~\ref{line:sh-initC}--\ref{line:sh-finishC}.
To see this, first consider some~$i \in [\tau]$ for which $p$ obtains $a_j>0$ points due to each voter in~$V_i$ in $\E_{C'}$. Let $Q_1,\dots,Q_{j-1} \in \Q$ be the equivalence classes that contain the pairs~$(i,1),\dots,(i,j-1)$, respectively. First, if $Q_h=Q_{\opp}$ for some~$h \in [j-1]$, then $\opp$ precedes~$p$ in the preferences of voters in~$V_i$, by the definition of~$Q_p$, $Q_{\opp}$, and the guesses within  iteration~$\iota$. Furthermore, for each $h \in [j-1]$, matching~$M$ contains some edge~$(P_h,Q_h)$ covering~$Q_h$, and thus $C_M$ contains a candidate~$c_h$ nominated by~$P_h$ that is well placed for~$Q_h$. By definition, this means that $c_h$ precedes~$p$ in~$\E_{C_M}$. Since this holds for each $h \in [j-1]$, we get that $\scr[\E_{C_M}](p)\leq \scr[\E_{C'}](p)$. 

Analogously, consider some $i \in [\tau]$ for which $\opp$ obtains $a_j>0$ points due to each voter in~$V_i$  in $\E_{C'}$. Let $Q_1,\dots,Q_{j-1} \in \Q$ be the equivalence classes that contain the pairs~$(i,1),\dots,(i,j-1)$, respectively. On the one hand, if some candidate~$c$ is nominated by some secure party in~$\E_{C_M}$, then $c$ is safe, meaning that voters in~$V_i$ prefer~$\opp$ to~$c$ (because $(i,j) \in Q_{\opp}$). On the other hand, if $c$ is nominated by some party~$\widetilde{P}$ that is covered by some edge~$(\widetilde{P},Q)$ of~$M$ for some~$Q \in \hat{\Q}$, then either $Q \in \{Q_1,\dots,Q_{j-1}\}$ or $c$ does not precede $\opp$ in the preferences of voters in~$V_i$, because $c$ is well placed for~$Q$. This means that $\opp$ is preceded by at most $j-1$ nominees in~$\E_{C_M}$, implying $\scr[\E_{C_M}](\opp)\geq \scr[\E_{C'}](\opp)$.

By the previous two paragraphs, we obtain \[\scr[\E_{C_M}](p)\leq \scr[\E_{C'}](p) < \scr[\E_{C'}](\opp) \leq \scr[\E_{C_M}](\opp),\] 
so $p$ is not a winner in~$\E_{C_M}$, as promised. %
Hence, Algorithm~\ref{alg:short-FPT} returns ``no'' on line~\ref{line:sh-returnyes} in iteration~$\iota$, proving the correctness of the algorithm.

\end{proof}

We next present an algorithm that solves \NPP{} for Veto-like scoring rules in FPT time when parameterized by~$\tau$.
\begin{theorem}\label{thm:FPT-veto}
    Let $\R$ be a Veto-like voting rule, based on a scoring vector of the form  $(a,\dots,a,a_1,a_2,\dots,a_\ell)$ for some constant ${\ell \geq 1}$ where $a>a_1$. Then 
  \NPP{} for~$\R$ is $\FPT$ when parameterized by~$\tau$, the number of voter types.
\end{theorem}

\begin{proof}
Let $(\E,\P,p)$ be the input instance of \NPP.
First, if $|\P| \leq \ell \tau$, then we consider $\R$ as an $\ell'$-short voting rule for $\ell'=|\P|\leq \tau \ell$ 
and apply Algorithm~\ref{alg:short-FPT}. As shown in the proof of Theorem~\ref{thm:FPT-short}, the running time is $O(%
|C|^{3.5} \cdot (\tau^2 \ell)^{\tau^2 \ell+4.5})$ which is fixed-parameter tractable for~$\tau$ because $\ell$ is a constant.

Second, if $|\P|>\ell \tau$, then there will be at least one nominated candidate in each reduced election~$\E'$ that achieves the maximum possible score of $|V| \cdot a$.
This means that $p$ is not a winner in~$\E'$ if and only if $p$ has score less than~$|V| \cdot a$, i.e., it is ranked among the $\ell$ least favorite nominees for some voter. 
To decide whether this is possible, it suffices the check %
whether 
\begin{equation}
\label{eq:vetolike-condition}
    |\{\widetilde{P} \in \P: c \succ_v p \text{ for some } c \in \widetilde{P}\} | \geq |\P|- \ell
\end{equation} 
for some voter~$v \in V$. 
Clearly, this can be checked in polynomial  time.
If Inequality~(\ref{eq:vetolike-condition}) holds for some $v \in V$, then we return ``no'', otherwise we return ``yes''. 

To see the correctness of this algorithm,
assume that Inequality~(\ref{eq:vetolike-condition})    holds for some voter $v \in V$. 
Then nominating~$p$ %
and nominating a candidate preferred to~$p$ for every other party where this is possible, we get that $p$ will be preceded by at least~$|\P|-\ell$ nominees in the preferences of~$v$ in the resulting election~$\E'$. Hence, $v$ allocates less than~$a$ points to~$p$, yielding $\scr[\E'](p)<|V| \cdot a$, which in turn implies that $p$ is not a winner in~$\E'$.
Assume now that Inequality~(\ref{eq:vetolike-condition}) does not hold for any  voter~$v \in V$. Then irrespective of the nominations from the parties, for each voter~$v \in V$ there will exist less than $|\P|-\ell$ parties whose nominee is preferred by~$v$ to~$p$. Thus, $p$ receives $a$ points from each voter, yielding a total score of $|V| \cdot a$, the maximum obtainable score, which ensures that $p$ is a winner in the resulting election.
\end{proof}

\section{Results for Condorcet-Consistent Voting Rules}
\label{sec:condorcet}

In this section we deal with three Condorcet-consistent voting rules. We show that \NPP{} is polynomial-time solvable for Copeland$^\alpha$ as well as for Maximin. By contrast, \NPP{} for Ranked Pairs is computationally hard even for a constant number of voters.

\begin{theorem}
\label{thm:Copeland} 
 For each $\alpha \in [0,1]$, \NPP{} for Copeland$^\alpha$ is polynomial-time solvable.
\end{theorem}

\begin{proof}
We propose Algorithm~\myref{alg:Copeland} that solves  \NPP\ for Copeland$^\alpha$ voting in polynomial time.
Let $(\E,\P,p)$ be our input instance with election $\E=(C,V,\{\succ_v\}_{v \in V})$ and party~$P$ containing~$p$. 

\begin{varalgorithm}{NP-Copeland}
\caption{Solving \NPP\ for Copeland$^\alpha$ for some $\alpha \in [0,1]$.}
\label{alg:Copeland}
\begin{algorithmic}[1]
\Require{An instance~$(\E,\P,p)$ of \NPP{} with candidate set~$C$ and $p \in P \in \P$.}
\ForAll{%
$p' \in C \setminus P$}
    \State Let $P'$ be the party in~$\P$ containing~$p'$.
    \ForAll{$c \in C \setminus (P \cup P')$} compute $\Delta_c(p',p)$. %
	\EndFor
    
    \State Set $\widetilde{C}=\{p,p'\}$.
    \ForAll{$\widetilde{P} \in \P \setminus \{P,P'\}$}
        \State Add a candidate $c(\widetilde{P})\in \arg\max_{c \in \widetilde{P}} \Delta_c(p'p)$ 
        to~$\widetilde{C}$.
    \EndFor
     \If{$p$ is not a winner in $\E_{\widetilde{C}}$}  {\bf return} ``no''. 
          
     \EndIf
\EndFor
\State {\bf return} ``yes''. \end{algorithmic}
\end{varalgorithm}

Assume that there exists a set~$C' \ni p$ of nominated candidates for which $p$ is \emph{not} a winner in the reduced election~$\E_{C'}$ over~$C'$. 
Algorithm~\myref{alg:Copeland} first guesses %
a candidate~$\opp \in  C' \setminus P$ such that $\Cpl[\E_{C'}](\opp)>\Cpl[\E_{C'}](p)$. 
Let $\oppParty $ be the party containing~$\opp$.

Recall that $\Cpl[\E_{C'}](c,c')$ for two candidates $c$ and~$c'$ is the score received by~$c$ resulting from the head-to-head comparison of~$c$ with~$c'$ in $\E_{C'}$, see its definition in Equation~(\ref{eq:define_Copeland}).
Let us now define
$\Delta(\opp,p)=\Cpl[\E_{C'}](\opp,p)-\Cpl[\E_{C'}](p,\opp)$; note that this value is the same for all reduced elections~$\E$ containing both $p$ and~$\opp$, and  
can be computed easily from the election $\E$. 
Then
the difference in the Copeland$^\alpha$ score of $\opp$ and~$p$ in~$\E_{C'}$ 
can be expressed as 
\begin{equation} 
\label{eq:Copeland-diff-score}
\Cpl[\E_{C'}](\opp)-\Cpl[\E_{C'}](p) = 
\Delta(\opp,p)+\sum_{c \in C' \setminus \{p,\opp\}} \Delta_c(\opp,p)
\end{equation}
where 
$\Delta_c(\opp,p)=\Cpl[\E_{C'}](\opp,c)-\Cpl[\E_{C'}](p,c)$ for each candidate $c \in C' \setminus \{p,\opp\}$ that is nominated in~$\E$. Hence,
the value 
$\Delta_c(\opp,p)$ reflects the difference resulting in the score of~$\opp$ and~$p$ from their comparison with some candidate $c \in C'\setminus \{p,\opp\}$. Notice that $\Delta_c(\opp,p)$ is the same for all reduced elections of~$\E$ containing $p$, $\opp$, and $c$, and can be calculated according to the following cases:
\begin{itemize}
    \item[$(i)$] if $\opp$ defeats $c$ and $c$ defeats  $p$, then  $\Delta_c(\opp,p)=1$;
    \item[$(ii)$] if $\opp$ defeats $c$ and $p$ is tied with $c$, then  $\Delta_c(\opp,p)=1-\alpha$;
    \item[$(iii)$] if $\opp$ is tied with $c$ and $p$ is defeated by $c$, then  $\Delta_c(\opp,p)=\alpha$;
     \item[$(iv)$] if $\opp$ is tied with $c$ and $p$ defeats $c$, then  $\Delta_c(\opp,p)=\alpha-1$;
     \item[$(v)$] if $\opp$ is defeated by $c$ and $p$ is tied with $c$, then  $\Delta_c(\opp,p)=-\alpha$;
    \item[$(vi)$] if $\opp$ is defeated by $c$ and $p$ defeats $c$, then  $\Delta_c(\opp,p)=-1$;
    \item[$(vii)$] if both $p$ and $\opp$ defeat $c$, or are both tied with $c$, or are both defeated by $c$, then $\Delta_c(\opp,p)=0$.
\end{itemize}

After guessing %
$\opp$,  Algorithm~\myref{alg:Copeland} nominates a candidate~$c_{\widetilde{P}} \in \widetilde{P}$ 
for each party~$\widetilde{P} \in \P \setminus \{P,\oppParty \}$ 
maximizing~$\Delta_c(\opp,p)$ over~$\widetilde{P}$.
Finally,  the algorithm computes %
$\Cpl[\E_{\widetilde{C}}](\opp)-\Cpl[\E_{\widetilde{C}}f](p)$ according to Equation~(\ref{eq:Copeland-diff-score}) for the obtained reduced election~$\E_{\widetilde{C}}$; if this value is positive, it returns ``no.'' 
If the algorithm has explored all possible guesses for~$\opp$ without returning ``no'', then
it returns ``yes.''

Let us prove the correctness of Algorithm~\myref{alg:Copeland}.
It is clear that whenever Algorithm~\myref{alg:Copeland} returns ``no'', then it does so correctly, because  $\Cpl[\E_{\widetilde{C}}](\opp)> \Cpl[\E_{\widetilde{C}}](p)$ and so candidate $p$ is not a winner of the reduced election~$\E_{\widetilde{C}}$ over the candidate set~$\widetilde{C}$ constructed by the algorithm. %

Conversely, assume that a candidate $\opp$ %
is a winner of some reduced election $\E_{C'}$  over a set~$C' \ni p$ of nominated candidates in which $p$ is not a winner. Let $\oppParty $ be the party containing~$\opp$, and let $c_{\widetilde{P}}$ denote the nominee of some party~$\widetilde{P}$ other than $P$ or~$\oppParty $ in~$\E_{C'}$.
Since $\opp$ is a winner in~$\E_{C'}$ but $p$ is not, we have $\Cpl[\E_{C'}](\opp)> \Cpl[\E_{C'}](p)$.
Using Equation~(\ref{eq:Copeland-diff-score}) and the algorithm's choice for the nominated candidates, we get that
\begin{align*}
\mathsf{Cpl}_{\E_{\widetilde{C}}}^\alpha&(\opp)- \mathsf{Cpl}_{\E_{\widetilde{C}}}^\alpha(p)
=
\Delta(\opp,p)+\sum_{\widetilde{P} \in \widetilde{\P}} \max \{\Delta_c(\opp,p):c \in \widetilde{P}\}
\\
& \!\!\!\! \geq 
\Delta(\opp,p)+\sum_{\widetilde{P} \in \widetilde{\P}} \Delta_{c_{\widetilde{P}}}(\opp,p)
= \Cpl[\E_{C'}](\opp)-\Cpl[\E_{C'}](p)>0
\end{align*}
for $\widetilde{\P}=\P \setminus \{P,\oppParty \}$.
Thus, Algorithm~\myref{alg:Copeland} returns ``no.'' 

To evaluate the computational complexity of Algorithm~\myref{alg:Copeland},
note that there are at most $|C|%
$ possibilities to choose%
~$\opp$. The values~$\Delta_c(\opp,p)$ can be computed in time $O(|V| \cdot |C|)$ 
which also suffices to find the nominated candidates and compute the resulting Copeland$^\alpha$ scores of~$p$ and~$\opp$.
Hence the total running time of Algorithm~\myref{alg:Copeland} is $O(|C|^2 \cdot |V|)$.
\end{proof}

\begin{theorem}
\label{thm:Maximin} 
{\sc Necessary President} for Maximin is polynomial-time solvable.
\end{theorem}
\begin{proof}
We propose an algorithm that solves  \NPP\ for Maximin in polynomial time; see Algorithm~\ref{alg:Maximin} for a pseudocode. 
Let $(\E,\P,p)$ be our input instance with election $\E=(C,V,\{\succ_v\}_{v \in V})$ and party~$P$ containing~$p$. 

Assume that there exists a reduced election~$\E'$ where $p$ is nominated but is not a winner. Then there exists a nominee~$\opp$ in~$\E'$ whose Maximin score exceeds the Maximin score of~$p$. Let $\oppParty $ be the party containing $\opp$.
Algorithm \ref{alg:Maximin} first guesses candidate~$\opp$ as well as a candidate~$\hat{c} \in C \setminus (P \cup \oppParty )$ that determines the Maximin score of~$p$ in~$\E'$, i.e., for which $\MM_{\E'}(p)=N(p,\hat{c})$. Let $\hat{s}=N(p,\hat{c})$ and $\hat{P}$ be the party containing $\hat{c}$.\footnote{Notice that the value $N_{\E'}(c,c')$ for any pair of candidates $c,c'$ is the same irrespective of other nominations and so also of the reduced election, therefore in the rest of this \new{paper we shall omit the index indicating the (reduced) election.}}%

First, the algorithm checks whether its guesses are valid in the sense that hold %
$N(\opp,\hat{c})>\hat{s}$ and $N(\opp,p)>\hat{s}$ (conditions necessary for $\MM_\E(\opp)>\hat{s}$). 
Next, the algorithm searches for a suitable nominee~$c_{\widetilde{P}}$ for each party~$\widetilde{P} \in \mathcal{P} \setminus \{P,\oppParty ,\hat{P}\}$ that satisfies 
$N(\opp,c_{\widetilde{P}}) > \hat{s}$.
If such a candidate is found for each party other than~$P$, $\oppParty $, and~$\hat{P}$, then \ref {alg:Maximin} returns ``no.''
If all guesses are exhausted but the algorithm has not output ``no'', then it returns ``yes.''

\begin{varalgorithm}{NP-Maximin}
\caption{Solving \NPP\ for Maximin.}
\label{alg:Maximin}
\begin{algorithmic}[1]
\Require{An instance~$(\E,\P,p)$ of \NPP\ with candidate set~$C$.}
\ForAll{%
$\opp \in C \setminus P$}
    \State Let $\oppParty $ be the party in~$\P$ containing~$\opp$.
    \ForAll{$\hat{c} \in C \setminus (P \cup \oppParty )$}
        \State Let $\hat{s}=N(p,\hat{c})$ and let $\hat{P}$ be the party containing $\hat{c}$.
        \If{$N(\opp,\hat{c})>\hat{s}$ and $N(\opp,p)>\hat{s}$}\label{line:Mx-valid}
            \If{$\forall \widetilde{P} \in \mathcal{P} \setminus \{P, \oppParty ,\hat{P}\} \,\, \exists c \in \widetilde{P}: N(\opp,c)>\hat{s}$}\label{line:Mx-nominees} 
                \State {\bf return} ``no''.\label{line:Mx-yes}
            \EndIf
        \EndIf
	\EndFor
\EndFor
\State {\bf return} ``yes''. \end{algorithmic}
\end{varalgorithm}

To show the correctness of \ref{alg:Maximin}, observe first that assuming correct guesses, 
the conditions checked on line~\ref{line:Mx-valid} must hold for $p$, $\opp$, and $\hat{c}$, because by our definitions we have 
\[\hat{s}=N(p,\hat{c})=\MM_{\E'}(p)<\MM_{\E'}(\opp)
\leq \min\{N(\opp,p),N(\opp,\hat{c})\}.\]

Similarly, for each nominee~$c$ in~$\E'$ other than these three candidates, we know $\hat{s}=\MM_{\E'}(p)<\MM_{\E'}(\opp)\leq N(\opp,c)$. Therefore, the algorithm will find on line~\ref{line:Mx-nominees} that there exists a candidate $c$ satisfying the requirement $N(\opp,c)>\hat{s}$ in each party not containing $p,\opp$, or $\hat{c}$. Hence, \ref{alg:Maximin} will return ``no'' on line~\ref{line:Mx-yes}.

For the other direction, assume that \ref{alg:Maximin} returns ``no'' in some iteration. Consider the reduced election obtained where $P$, $\oppParty $, and $\hat{P}$ nominate the candidates $p$, $\opp$, and $\hat{c}$, respectively, guessed in this iteration, 
while each remaining party $\widetilde{P}$ nominates a candidate $c_{\widetilde{P}}$ that satisfies $N(\opp,c_{\widetilde{P}})>\hat{s}$.
Note that such candidates exist because the algorithm found the condition on line~\ref{line:Mx-nominees} to hold in the  iteration when it returned ``no.'' 
Hence, this method indeed yields a reduced election~$\E_{C'}$ over some candidate set~$C'$. Clearly, our assumptions on the nominees  ensure $\MM_{\E_{C'}}(\opp)=\min_{c \in C' \setminus \{\opp\}}N(\opp,c)>\hat{s}=N(p,\hat{c}) \geq \MM_{\E_{C'}}(p)$, and therefore $p$ is not a winner in~${\E_{C'}}$.

To estimate the computational complexity of the algorithm we first note that there are at most $|C|^2$ possible guesses for candidates~$\opp$ and~$\hat{c}$. To compute the scores $N_{\E}(c,c')$ for each candidate pair $c$ and $c'$ in $C$ requires $O(|C|^2\cdot |V|)$ steps. To check the condition on line~\ref{line:Mx-nominees} for some fixed guess takes $O(|C|)$ time.
Hence, the total running time of \ref{alg:Maximin} is $O(|C|^2(|C|+|V|))$.
\end{proof}

Finally, we show that, unlike for Copeland$^\alpha$ and Maximin, \NPP{} is computationally hard for Ranked Pairs, even for a constant number of voters.
\ifshort
The reduction showing $\coNP$-hardness in
Theorem~\ref{thm:RP_coNPcomplete_W1hard} is from \textsc{$3$-SAT}, while 
$\mathsf{W}[1]$-hardness is obtained by a reduction from the %
\textsc{Multicolored Clique}~\cite{pietrzak-multicolored-2003} problem.
\else
Theorem~\ref{thm:RPcoNP} relies on a reduction from \textsc{$3$-SAT}, while the proof of Theorem~\ref{thm:RP_W1hard} gives a parameterized reduction from the $\mathsf{W}[1]$-hard \textsc{Multicolored Clique} problem~\cite{pietrzak-multicolored-2003}.
\fi

\begin{theorem}
\label{thm:RPcoNP}
    \NPP\ for Ranked Pairs is $\coNP$-complete even if the maximum party size is $\maxsize=2$ and the number of voters is $|V|=12$.
\end{theorem}
\begin{proof}
We present a reduction from the $\NP$-complete problem 
\new{\sc(2,2)-E3-SAT.}
Let our input formula be $\varphi = \bigwedge_{C \in \mathcal{C}} C$ over a set~$X$ of variables.
We use the notation $X=\{x_1,\dots, x_r\}$ and 
$\C=\{C_1,\dots, C_q\}$.

To define an instance of \NPP\ for Ranked Pairs, we let the candidate set be 
\[\
\{p,\opp\} \cup \{x,\ol{x}:x \in X \} \cup \{C_i^j, C_i^{\neg j}:C_i \in \C, j \in [3]\}
\]
with $p$ as the distinguished candidate.
We define $P=\{p\}$ and~$\oppParty =\{\opp\}$ as singleton parties, and we also add the parties $P_x=\{x,\ol{x}\}$ for each variable $x \in X$ and three parties $P^j_{C_i}=\{C_i^j,C_i^{\neg j}\}$, $j \in [3]$, for each clause $C_i \in  \C$. 
Intuitively, the nominee of party~$P_x$ corresponds to selecting the truth assignment for variable~$x \in X$, while nominating a candidate $C_i^j$ for some clause~$C_i \in \C$ and $j \in [3]$ corresponds to setting the $j^{\textrm{th}}$ literal in clause~$C_i$ to true.
Note that the maximum party size is indeed $\maxsize=2$.

We define the set of voters as $Y \cup Z$ where 
$Y=\{y_h,y'_h:h \in [4]\}$ and 
$Z=\{z_h,z'_h:h \in [2]\}$, so the number of voters is $n=12$.
To define the preferences of voters', let us define the following sets of candidates: $L=\{x,\ol{x}:x \in X\}$ contains all literals,
and we define $\C^+=\{C_i^j:C_i \in \C, j \in [3]\}$
and $\C^-=\{C_i^{\neg j}:C_i \in \C, j \in [3]\}$.
Additionally, for each literal $\ell \in L$, we let \[
A(\ell)=\{C_i^j:\ol{\ell} \textrm{ is the }j^\textrm{th} \textrm{ literal in $C_i$}\}.\]
Notice that the sets~$A(\ell)$, $\ell \in L$, yield a partitioning of the set $\{C_i^j:j \in [3], C_i \in \C\}$ of candidates.

To define the preferences of the voters, we first need some additional notation. 
We fix an arbitrary ordering over the set of all candidates defined, and for any subset~$S$ of the candidates, we let $\ora{S}$ denote the (strict) preference ordering over~$S$ corresponding to this fixed order, and similarly, we let $\ola{S}$ denote the reverse preference ordering over~$S$. 
Then preferences of voters are as follows:

\begin{align*}
    y_h \text{ for }& h \in [4] : \\
    & p \succ \ora{L}  \succ C_1^{\neg 1} \succ C_1^{\neg 2}  \succ C_2^{\neg 1} \succ C_2^{\neg 2} \succ \shortdots \succ C_q^{\neg 1} \succ C_q^{\neg 2} \\
    & \phantom{p}
    \succ C_1^{\neg 3} \succ C_2^{\neg 3}\succ \shortdots \succ C_q^{\neg 3} \succ  \ora{\C^+} \succ \opp %
    \\
    y'_h \text{ for }& h \in [2] : \\ 
    & \ola{\C^+} \succ \opp \succ C_q^{\neg 2} \succ C_q^{\neg 3}  \succ C_{q-1}^{\neg 2} \succ C_{q-1}^{\neg 3} \succ \shortdots \succ C_1^{\neg 2} \succ C_1^{\neg 3} \\
    & \phantom{\ola{\C^+}}
    \succ p \succ C_q^{\neg 1} \succ C_{q-1}^{\neg 1} \succ\shortdots \succ C_1^{\neg 1} \succ \ola{L}  %
    \\
    y'_h \text{ for }&  h \in \{3,4\}: \\
    & \ola{\C^+} \succ  C_q^{\neg 3} \succ C_{q-1}^{\neg 3} \succ \shortdots \succ C_1^{\neg 3} \succ \opp \succ C_q^{\neg 1} \succ C_q^{\neg 2} \\
    & \phantom{:\ola{\C^+}}
    \succ C_{q-1}^{\neg 1} \succ C_{q-1}^{\neg 2} \succ \shortdots \succ C_1^{\neg 1} \succ C_1^{\neg 2} \succ p \succ \ola{L}   %
    \\
    z_h \text{ for }& h \in [2] : \\ 
    & x_1 \succ A(x_1) \succ  x_2 \succ A(x_2) \succ \shortdots \succ  x_r \succ A(x_r) \succ \ol{x}_1 \succ A(\ol{x}_1) \\
    & \phantom{:x_1} \succ 
     \ol{x}_2 \succ A(\ol{x}_2) \succ \shortdots \succ \ol{x}_r \succ A(\ol{x}_r) \succ \opp \succ p \succ \ora{\C^{-}} 
     \\
     z'_1 : \phantom{fo}& 
     \ola{\C^{-}}  \succ p \succ \opp \\
     & \phantom{:\ola{\C^{-}}} \succ 
     \ol{x}_r \succ A(\ol{x}_r) \succ \ol{x}_{r-1} \succ A(\ol{x}_{r-1})
     \succ \shortdots \succ \ol{x}_1 \succ A(\ol{x}_1) \\
     & \phantom{:\ola{\C^{-}}} \succ 
     x_r \succ A(x_r) \succ  x_{r-1} \succ A(x_{r-1}) \succ \shortdots \succ x_1 \succ A(x_1); 
     \\
     z'_2 : \phantom{fo}&  
     \ola{\C^{-}} \succ \opp \succ p 
     \\ 
     & \phantom{:\ola{\C^{-}}i} \succ 
     \ol{x}_r \succ A(\ol{x}_r) \succ \ol{x}_{r-1} \succ A(\ol{x}_{r-1})
     \succ \shortdots \succ \ol{x}_1 \succ A(\ol{x}_1) \\
     & \phantom{:\ola{\C^{-}}i} \succ 
     x_r \succ A(x_r) \succ  x_{r-1} \succ A(x_{r-1}) \succ \shortdots \succ x_1 \succ A(x_1).
\end{align*}
This finishes the construction of our instance~$I$ of \NPP. We are going to show that $I$ admits a reduced election where $p$ is \emph{not} a winner if and only if $\varphi$ is satisfiable.

Let us start with computing the values $N(c,c')$ for each pair of distinct candidates~$c$ and~$c'$ in~$I$. The following can be observed directly from the preferences of voters:
\begin{itemize}
    \item $N(p,\ell)=10$ for each literal $\ell \in L$;
    \item $N(C_i^j,\opp)=10$ for each clause $C_i \in \C$ and $j \in [3]$;
    \item $N(p,C_i^{\neg 1})=N(C_i^{\neg 1},C_i^{\neg 2})=N(C_i^{\neg 2},C_i^{\neg 3})=N(C_i^{\neg 3},\opp)=8$ for each clause $C_i \in \C$;
    \item $N(\opp,p)=7$;
    \item $N(a,a')=6$ for each pair $(a,a')$ of distinct candidates not listed above and satisfying $(a,a') \notin (L \times \C^+) \cup (\C^+ \times L)$. 
\end{itemize}

Consider now a literal~$\ell \in L$ and  some candidate $C_i^j \in \C^+$. It is clear that exactly four voters in~$Y$ prefer~$\ell$ to~$C_i^j$. Moreover, all four voters in~$Z$ prefer~$\ell$ to~$C_i^j$ if $C_i^j \in A(\ell)$; otherwise, there are only two such voters in~$Z$. Hence, we get that
\begin{equation}
\label{eq:LvsC}
    N(\ell, C_i^j)=\left\{ 
    \begin{array}{ll}
         8, & \textrm{ if $C_i^j \in A(\ell)$,} \\
         6, & \textrm{ if $C_i^j \notin A(\ell)$.}
    \end{array}
    \right.
\end{equation}

Let $F$ denote the set of edges added to the acyclic subgraph~$D$ of the majority graph constructed during the winner determination process for Ranked Pairs in a reduced election~$\E_{C'}$ over some candidate set~$C'$; see Figure~\ref{fig:RP_SAT}. 
Notice that the arc set 
\begin{align*}
    F_0 = & \left(\{p\} \times L) \cup (\C^+ \times \{\opp\} \right) \\
    & \cup \{(p,C_i^{\neg 1}), (C_i^{\neg 1},C_i^{\neg 2}), (C_i^{\neg 2},C_i^{\neg 3}),(C_i^{\neg 3},\opp):C_i \in \C\} 
\end{align*}
is acyclic. In fact, even $F_0 \cup (L \times \C^+)$ is acyclic.
Since---by our observations on the pairwise comparisons between candidates---this set contains all arcs with weight at least~$8$ in the weighted majority graph of the instance, we obtain that 
\[F_0 \cap (C' \times C') \subseteq F \subseteq F_0 \cup (L \times \C^+) \cup \{(\opp,p)\}.\]

\begin{figure}[thb]
\centering
\includegraphics[width=\columnwidth]{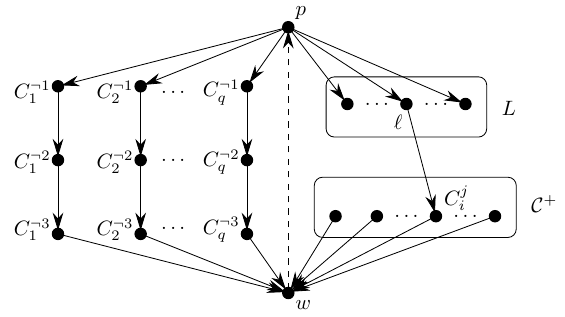}
\caption{Illustration for the acyclic subgraph~$D$ formed during the winner determination procedure of Ranked Pairs for the instance constructed in the proof of Theorem~\ref{thm:RPcoNP}. Solid arcs are always included in~$D$, while  the arc~$(\opp,p)$, shown as dashed arrows, may or may not be present in~$D$. The figure assumes $C_i^j \in A(\ell)$.}
\label{fig:RP_SAT}
\Description{Illustration for the acyclic subgraph formed during the winner determination procedure of Ranked Pairs for the instance constructed in the proof of Theorem~\ref{thm:RPcoNP}.}
\end{figure}

Notice furthermore that the only possible incoming arc for~$p$ is $(\opp,p)$, which implies that $p$ is \emph{not} a winner in the reduced election~$\E_{C'}$ if and only if the arc~$(\opp,p)$ is present in~$F$. We are going to show that this happens if and only if $\varphi$ is satisfiable. 

\medskip
First, assume that $(\opp,p) \in F$. This clearly implies that there can be no path from~$p$ to~$\opp$ in~$D$, so in particular, 
\begin{itemize}
    \item[(a)] for each $C_i \in \C$, there is some $j \in [3]$ for which $C_i^{\neg j} \notin C'$;
    \item[(b)] $F$ contains no arc from $L \times \C^+$. 
\end{itemize}

Note that the nominees in~$\E_{C'}$ naturally determine a truth assignment, since each party~$P_x=\{x,\ol{x}\}$, $x \in X$, has exactly one nominee in~$\E_{C'}$, i.e., $|C' \cap P_x|=1$. Let~$\alpha_{C'}$ denote the truth assignment that sets to true exactly the literals in~$C' \cap L$.

Consider now some clause~$C_i \in \C$.
By (a), there exists some $j \in [3]$ such that the nominee of party~$P_{C_i}^j$ is~$C_i^j$. Let $\ell$ be the $j^{\textrm{th}}$ literal in~$C_i$; then $C_i^j \in A(\ol{\ell})$ by definition. Observe that  $\ol{\ell} \notin C'$, as otherwise the arc $(\ol{\ell},C_i^j)$ has both endpoints in~$C'$ and is thus present in~$F$, contradicting~(b).
Since $\ol{\ell} \notin C'$ implies $\ell \in C'$, we obtain that $\ell$ is set to true in the truth assignment~$\alpha_{C'}$, and thus the clause $C_i$ is satisfied. As this holds for all clauses in~$\C$, we obtain that $\alpha_{C'}$ satisfies~$\varphi$.

\medskip
Second, assume now that $\alpha$ is a satisfying truth assignment for~$\varphi$. 
Consider the reduced election $\E'$ where each party~$P_x$, $x \in X$, nominates the literal that is set to true in~$\alpha$, and each party $P_{C_i}^j$ nominates $C_i^j$ if and only if the $j$-th literal in clause~$C_i$ is set to true in~$\alpha$.

It is straightforward to verify that $p$ is \emph{not} a winner in~$\E'$. To see this, consider the acyclic subgraph~$D$ of the majority graph constructed during the winner determination process for Ranked Pairs in~$\E'$. We need to show that the subgraph of~$D$ that contains only arcs with weight at least~$8$ has no path from~$p$ to~$\opp$, and thus the arc~$(\opp,p)$ gets added to~$D$. 
First, notice that since~$\alpha$ satisfies~$\varphi$, for each clause~$C_i \in \C$ there exists some index~$j$ for which $C_i^j$ is nominated in~$\E'$, and thus $C_i^{\neg j}$ is \emph{not} present in~$\E'$. Thus, there is no path from~$p$ to~$\opp$ in~$D$ going through candidates in~$\C^-$. Second, observe that $D$ contains no arc in~$L \times \C^+$, because if some candidate~$C_i^j$ is a nominee in~$\E'$, then the literal~$\ell$ for which $C_i^j \in A(\ell)$ must be set to false by $\alpha$, and hence is not nominated in~$\E'$.
This shows that $(\opp,p)$ is an arc in~$D$ and hence $p$ is not a winner in~$\E'$.
\end{proof}

We finish with the following result, establishing the intractability of \NPP{} for Ranked Pairs even when the number of parties is a parameter and the number of voters is 20.

\new{In the proof we present a parameterized reduction from the $\mathsf{W}[1]$-hard problem \textsc{Multicolored Clique}~\cite{pietrzak-multicolored-2003}, defined as follows.} %

\medskip
\noindent
\begin{minipage}{\columnwidth}
\begin{framed}
\noindent {\bf Problem }{\sc Multicolored Clique}

\noindent {\bf Instance:} A graph $G=(U,E)$  with its vertex set~$U$ partitioned into $k$ independent sets $U_1, U_2,\dots, U_k$, and an integer $k$.

\noindent {\bf Question:}
Does there exist a \emph{multicolored clique}, that is, a clique of size~$k$ containing a vertex from each set~$U_i$, $i \in [k]$?
    
\end{framed}
\end{minipage}

\begin{theorem}
\label{thm:RP_W1hard}
    \NPP\ for Ranked Pairs is %
    $\mathsf{W}[1]$-hard with respect to parameter~$t$ denoting the number of parties, even if the number of voters is $|V|=20$.
\end{theorem}

\begin{proof}

   \new{Let our instance given for \textsc{Multicolored Clique} be a graph $G=(U,E)$, with the vertex set~$U$ partitioned into $k$ sets $U_1, U_2,\dots, U_k$, and an integer $k$.}
   We assume w.l.o.g.\ that each set $U_i$ contains the same number of vertices, say~$r$, so that we may denote the vertices in~$U_i$ as $u_i^1, u_i^2, \dots, u_i^r$ for each $i \in [k]$. 
    
    For two distinct integers~$i$ and~$j$ in~$[k]$, let $E_{\{i,j\}}$ denote the set of edges in~$G$ with one endpoint in~$U_i$ and one endpoint in~$U_j$; we will additionally use the notation $E_{i,>i}=\bigcup_{h:i<h\leq k} E_{\{i,h\}}$
    and $E_{i,<i}=\bigcup_{h:1\leq h<i} E_{\{i,h\}}$.
    Furthermore, for a set $F\subseteq E$ of edges in~$G$ and a vertex~$u \in U$, we let $F(u)$ denote the set of edges incident to~$u$ among those in~$F$. 

    To define an instance of \NPP, let us first define its candidate set as $U \cup E \cup \{p,\opp\}$ where $p$ and~$\opp$ are newly created candidates, with $p$ being the distinguished candidate, contained in party $P=\{p\}$; we also set $\oppParty =\{\opp\}$ as a singleton party. Additionally, for each $i \in [k]$ we set $V_i$ as a party, and similarly, for each $i,j \in [k]$ with $i \neq j$, we set $E_{\{i,j\}}$ as a party. This yields a set~$\mathcal{P}$ of $t=\binom{k}{2}+k+2$ parties.
    We will say that parties $U_i$ and $E_{\{j,\ell\}}$ are \emph{linked} if $i \in \{j,\ell\}$.
    
    We define the set of voters as  $V=X \cup Y \cup Z$ where $X=\{x_h,x'_h:h \in [4]\}$, $Y=\{y_h,y'_h:h \in [4]\}$, and $Z=\{z_h,z'_h:h \in [2]\}$; then the number of voters is $n=20$. 
    Again, as in the proof of Theorem~\ref{thm:RPcoNP}, we fix an arbitrary ordering over the set~$C$ of all candidates defined, and 
    let $\ora{C'}$ and~$\ola{C'}$ denote the preference list over a set~$C' \subseteq C$ of candidates determined by this ordering and its reverse, respectively.
    
    We first define the following partial preference orders which we will use as building blocks in the preference profile we create:
    \begin{align*}
        A_i &:= \ora{E_{i,>i}(u_i^1)} \succ u_i^1 \succ \ora{E_{i,>i}(u_i^2)} \succ u_i^2 
        \succ \shortdots \succ 
        \ora{E_{i,>i}(u_i^r)} \succ u_i^r; \\
        A'_i &:= \ola{E_{i,>i}(u_i^r)} \succ u_i^r \succ \ola{E_{i,>i}(u_i^{r-1})} \succ u_i^{r-1} 
        \succ \shortdots \succ 
        \ola{E_{i,>i}(u_i^1)} \succ u_i^1; \\[4pt]
        B_i &:= \ora{E_{i,<i}(u_i^1)} \succ u_i^1 \succ \ora{E_{i,<i}(u_i^2)} \succ u_i^2 
        \succ \shortdots \succ 
        \ora{E_{i,<i}(u_i^r)} \succ u_i^r; \\
        B   '_i &:= \ola{E_{i,<i}(u_i^r)} \succ u_i^r \succ \ola{E_{i,<i}(u_i^{r-1})} \succ u_i^{r-1} 
        \succ \shortdots \succ 
        \ola{E_{i,<i}(u_i^1)} \succ u_i^1.
    \end{align*}
    Now we are ready to define the preferences of voters as follows:
    \begin{align*}
        x_1 &: p \succ \opp \succ A_1 \succ A_2 \succ \shortdots \succ A_k; \\
        x_2 &: \opp \succ p \succ A_1 \succ A_2 \succ \shortdots \succ A_k; \\
        x'_h &\text{ for } h \in \{1,2\}: %
        A'_k \succ A'_{k-1} \succ \shortdots \succ A'_1 \succ  \opp \succ p %
        \\
        x_h &\text{ for } h \in \{3,4\}:  p \succ \opp \succ B_1 \succ B_2 \succ \shortdots \succ B_k %
        \\
        x'_h &\text{ for } h \in \{3,4\}: B'_k \succ B'_{k-1} \succ \shortdots \succ B'_1 \succ  \opp \succ p %
        \\
        y_h & \text{ for } h \in \{1,2\}: \\
        & p \succ \opp \succ \ora{V_1} \succ\ora{E_{1,>1}} \succ \ora{V_2},\ora{E_{2,>2}} \succ \shortdots \succ \ora{V_{k-1}} \succ \ora{E_{k-1,>k-1}} \succ \ora{V_k} %
        \\
        y'_h & \text{ for } h \in \{1,2\}: \\
        & \!\!\!  \ola{V_k} \succ \ola{V_{k-1}} \succ \ola{E_{k-1,>k-1}} \succ \shortdots \succ \ola{V_2} \succ \ola{E_{2,>2}} \succ \ola{V_1} \succ \ola{E_{1,>1}}  \succ \opp  \succ p %
        \\
        y_h & \text{ for } h \in \{3,4\}: \\
        & p \succ \opp \succ \ora{V_1} \succ \ora{V_2},\ora{E_{2,<2}} \succ \shortdots \succ \ora{V_k} \succ \ora{E_{k,<k}} 
        \\
        y'_h & \text{ for } h \in \{3,4\}: \\
        & \!\!\! \ola{V_k} \succ \ora{E_{k,<k}} \succ \ola{V_{k-1}} \succ \ola{E_{k-1,<k-1}} \succ \shortdots \succ \ola{V_2} \succ \ola{E_{2,<2}} \succ \ola{V_1} \succ \opp \succ p 
        \\
        z_h & \text{ for } h \in \{1,2\}: p \succ \ora{V} \succ \ora{E} \succ  \opp %
        \\
        z'_h & \text{ for } h \in \{1,2\}: \ola{E} \succ  \opp \succ p \succ \ola{V}. 
    \end{align*}

This finishes our construction. The reduction presented is a polynomial-time reduction and also a parameterized one, since the number of parties $t=\binom{k}{2}+k+2$ depends only on the original parameter~$k$. 

\medskip
We claim that there exists a reduced election for the constructed instance of \NPP\ for \new{Ranked Pairs}
where $p$ is \emph{not} a winner 
if and only if $G$ admits a multicolored clique of size~$k$. 
This would prove the correctness of the reduction and, hence, 
the $\mathsf{W}[1]$-hardness of the problem when parameterized by~$t$, even under the condition that $n=20$.

\medskip
\def\sel{\alpha}
Consider an arbitrary reduced election.
We can observe the following facts directly from the preferences of voters: 
\begin{itemize}
    \item $N(p,u)=12$ for each $u \in U$;
    \item $N(e,\opp)=12$ for each $e \in E$;
    \item $N(\opp,p)=11$;
    \item $N(u,u')=10$ for each distinct $u$ and $u'$ in~$U$;
    \item $N(e,e')=10$ for each distinct $e$ and $e'$ in~$E$;    
    \item $N(\opp,u)=10$ for each $u \in U$;
    \item $N(p,e)=10$ for each $e \in E$.    
\end{itemize}

Let us consider now the pairwise comparison between some candidates~$u \in U$ and $e \in E$. 
First assume that $u$ and $e$ belong to parties that are not linked. In this case, they do not simultaneously appear in any of the blocks~$A_i$, $A'_i$, $B_i$, or~$B'_i$ for $i \in [k]$, and thus exactly half of the voters in~$X$ prefer~$u$ to~$e$. Moreover, $u$ is preferred to~$e$ by exactly half of the voters in~$Y$ as well; the same obviously holds  for voters in~$Z$. Hence, $u$ and~$e$ are tied in voters' preferences unless they belong to linked parties.

Assume now $u$ and $e$ belong to linked parties, that is, $u \in U_i$ and $e \in E_{\{i,j\}}$ for two distinct indices~$i,j \in [k]$.
Consider first voters in~$X$. Clearly, $u$ appears in each of~$A_i,A'_i,B_i,$ and~$B'_i$.
However, $e$ appears in~$A_i$ and~$A'_i$ if and only if $i<j$, while it appears in~$B_i$ and~$B'_i$ if and only if $i>j$. 
Moreover, if $i<j$, then $e$ is preferred to~$u$ by all voters in~$\{x_h,x'_h:h \in [2]\}$ if and only if $e \in E_{i,>i}(u)$, otherwise $e$ and~$u$ are tied in these four voters' preferences. 
Similarly, if $i>j$, then $e$ is preferred to~$u$ by all voters in~$\{x_h,x'_h:h \in \{3,4\}\}$ if and only if $e \in E_{i,<i}(u)$, otherwise $e$ and~$u$ are tied in these four voters' preferences. 
Therefore, we obtain that in either case (whether $i<j$ or $i>j$) exactly six voters from~$X$ prefer~$e$ to~$u$ if $e$ is incident to~$u$, otherwise they are tied in the preferences of voters in~$X$.

Consider now voters in~$Y$.
On the one hand, if $i<j$, then $u$ is preferred to~$e$ by all voters in~$Y$ except for voters~$y'_3$ and~$y'_4$. 
On the other hand, if $i>j$, then $u$ is preferred to~$e$ by all voters in~$Y$ except for voters~$y_1$ and~$y_2$. Hence, in either case there are six voters in~$Y$ preferring~$u$ to~$e$, and there are two more such voters in~$Z$. 

Summing all this up, we get that 
    \[ N(u,e)=\left\{
    \begin{array}{@{}l@{\hspace{4pt}}l@{}}
        12, & \text{if $e$ is not incident to~$u$ but their parties are linked,}\\
        10, & \text{otherwise.} \\
    \end{array} 
    \right.
    \]

\begin{figure}[th]
\centering
\includegraphics[scale=1]{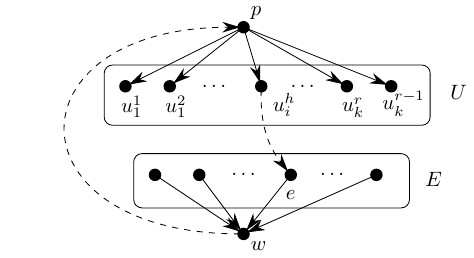}
\caption{Illustration for the acyclic subgraph~$D$ formed during the winner determination procedure of Ranked Pairs for the instance constructed in the proof of Theorem~\ref{thm:RP_W1hard}. Solid arcs are always included in~$D$, while arcs in $U \times E$ and the arc~$(\opp,p)$, shown as dashed arrows, may or may not be present in~$D$.}
\label{fig:RP_clique}
\Description{Illustration for the acyclic subgraph formed during the winner determination procedure of Ranked Pairs for the instance constructed in the proof of Theorem~\ref{thm:RP_W1hard}.}
\end{figure}
Let $F$ denote the set of edges added to the acyclic subgraph $D$ of the majority graph constructed during the winner determination process for Ranked Pairs in a reduced election~$\E_{C'}$ over some candidate set~$C'$; see Figure~\ref{fig:RP_clique}. By our observations on the pairwise comparisons between candidates, we know that $F$ contains all arcs in~$(\{p\} \times U) \cup (E \times \{\opp\})$ whose endpoints are both both candidates nominated in~$\E_{C'}$, i.e., candidates in~$C'$. That is, $F$ contains the arc set
\[F_0=\left((\{p\} \times U) %
    \cup (E \times \{\opp\}) \right) \cap (C' \times C').\]
Additionally, $F$ might also contain some arcs from $U \times E$ or possibly the arc~$(\opp,p)$, so we get
\begin{align*}
    F_0 \subseteq F %
    \subseteq F_0 \cup (U \times E) \cup \{(\opp,p)\}.
\end{align*}

Notice now that any arc~$(u,e) \in F$ creates a path from~$p$ to~$\opp$ in~$D$. 
In fact, there is a path from $p$ to~$\opp$ in~$D$ if and only if $F \cap (U \times E) \neq \emptyset$, which in turn happens if and only if $(\opp,p) \notin F$. 
Notice furthermore that the only possible incoming arc for~$p$ is $(\opp,p)$, which implies that $p$ is \emph{not} a winner in the reduced election~$\E_{C'}$ if and only if the arc $(\opp,p)$ is present in~$F$, or equivalently, if $F \cap (V \times E) = \emptyset$.

It remains to show that $F \cap (V \times E) = \emptyset$ holds if and only if the nominees of all parties except $P$ and~$\oppParty $ form the vertices and edges of a multicolored clique in~$G$, 
which us equivalent with the property that
for each two nominees $e \in E$ and $u \in U$ belonging to linked parties, $e$ is incident to~$u$ in~$G$.

Let $u_i$ denote the nominee of~$U_i$ for each $i \in [k]$.
Note that if the nominee~$e_{\{i,j\}}$ of some party~$E_{\{i,j\}}$ is \emph{not} incident to~$u_i$, then by $N(u_i,e_{\{i,j\}})=12$ the set~$F$ of arcs added when building the digraph~$D$ will contain the arc~$(u_i,e_{\{i,j\}})$; here we rely on the fact that the set $(\{p\} \times U) \cup (E \times \{\opp\}) \cup (U \times E)$ of arcs, containing all candidate pairs~$(c,c')$ with $N(c,c') \geq 12$, is acyclic. 
Conversely, if~$(u,e) \in U \times E$ is added to~$F$, then we must have $N(u,e)=12$ which happens only if $u$ and~$e$ belong to linked parties and $e$ is not incident to~$u$, as required; note
that here we rely on the fact that no arc in $U \times E$ is added to~$F$ after adding $(\opp,p)$, since the addition of~$(\opp,p)$ creates a path from each candidate in~$E$ to each candidate in~$U$, thereby preventing any later addition of some arc in $(U \times E)$ to~$F$.

This proves that it is possible to nominate candidates so that $p$ is not a winner in the resulting election if and only if $G$ admits a multicolored clique.
\end{proof}

\section{Conclusions and Outlook}
We explored the computational complexity of the  \NPP{} problem for several popular voting rules; together with previous results in \cite{cechlarova2023candidates}, \cite{schlotter2024}, and \cite{schlotter2025candidate}, our study offers a detailed picture of the computational tractability of problems faced by parties in the candidate nomination process preceding an election.

A possible direction for future research is 
to extend the existing tractability results for \NPP{} for Plurality under single-peaked~\cite{faliszewski2016} or single-crossing~\cite{misra2019parameterized} preferences to different voting rules or other structured domains.

As an interesting new topic, we propose to study candidate nomination problems for multiwinner elections. Suppose  that the goal of the election is to choose a committee consisting of $k$ members. In this case, it is natural to  assume that a party may nominate more than one candidate. What will its optimal strategy be if it wants to have at least one of its nominees in the committee or if it wishes to maximize the number of its nominees in the committee?

\begin{acks}
Ildik\'o Schlotter is supported by the Hungarian Academy of Sciences under its Momentum Programme (LP2021-2) and its J\'anos Bolyai Research Scholarship.
Katar{\'i}na Cechl{\'a}rov{\'a} is supported sby VEGA 1/0585/24 and
APVV-21-0369.
\end{acks}

\bibliographystyle{ACM-Reference-Format} 
\bibliography{articles}

\end{document}